\documentclass[11pt,a4paper]{article}
\usepackage[ansinew]{inputenc}
\usepackage{amsthm}
\usepackage{amssymb}
\usepackage{amsmath}
\usepackage{amsfonts}

\usepackage{epsfig}

\usepackage{graphics,graphicx}
\usepackage{color}

\newtheorem{theorem}[equation]{Theorem}
\newtheorem{corollary}[equation]{Corollary}

\newtheorem{proposition}[equation]{Proposition}

\theoremstyle{definition}

\newtheorem{remark}[equation]{Remark}

\numberwithin{equation}{section}

\newcommand{\N}{{\mathbb{N}}}

\newcommand{\B}{{\mathbb{B}}}

\begin{document}
%%%%%%%%%%%%%%%%%%%%%%%%%%%%%%%%%%%%%%%%%%%%%%%%%%%%%%%%%%%%%%%%%%%%%
%%%%%%%%%%%%%%%%%%%%%%%%%%%%%%%%%%%%%%%%%%%%%%%%%%%%%%%%%%%%%%%%%%%%%

\title{A Perron-Frobenius theory for block matrices associated to a multiplex network}

%\author{Miguel Romance}
%\author{Luis Sol\'a}
%\author{Julio Flores}
%\author{Esther~Garc\'{\i}a}
%\author{Alejandro~Garc\'{\i}a~del~Amo}
%\author{Regino~Criado}
%\affiliation{Department of Applied Mathematics, Rey Juan Carlos University, Madrid (Spain)}
%\affiliation{Center for Biomedical Technology (CTB), Technical University of Madrid (Spain)}

%\affiliation{Department of Applied Mathematics, Rey Juan Carlos University, Madrid (Spain)}
%\affiliation{Center for Biomedical Technology (CTB), Technical University of Madrid (Spain)}

%\affiliation{Department of Applied Mathematics, Rey Juan Carlos University, Madrid (Spain)}
%\affiliation{Center for Biomedical Technology (CTB), Technical University of Madrid (Spain)}

%\affiliation{Department of Applied Mathematics, Rey Juan Carlos University, Madrid (Spain)}
%\affiliation{Center for Biomedical Technology (CTB), Technical University of Madrid (Spain)}

%\affiliation{Department of Applied Mathematics, Rey Juan Carlos University, Madrid (Spain)}
%\affiliation{Center for Biomedical Technology (CTB), Technical University of Madrid (Spain)}
%%%%%%%%%%%%%%%%%%%%%%%%%%%%%%%%

\author{Miguel Romance, Luis Sol\'a, Julio Flores, Esther Garc\'{\i}a,\\ Alejandro Garc\'{\i}a del Amo and Regino Criado\\
\\
\small Departamento de Matem\'{a}tica Aplicada, CC. e Ing. de los Materiales \\ y \small Tec. Electr\'onica, URJC \\
\small Center for Biomedical Technology (CTB), UPM}
%\affiliation{Department of Applied Mathematics, Rey Juan Carlos University, Madrid %(Spain)}
%\affiliation{Center for Biomedical Technology (CTB), Technical University of Madrid (Spain)}
%Preliminary Version\\
%\\
%\small $^1$\,Departamento de Matem\'atica Aplicada} \\
%{\small ESCET - Universidad Rey Juan Carlos} \\
%{\small C/ Tulip\'an s/n 28933 M\'ostoles (Madrid), Spain.}\\
%}

%{\small C/ Tulip\'an s/n 28933 M\'ostoles (Madrid), Spain.}\\
%}
%\date{\today}
\date{}

\maketitle

%%%%%%%%%%%%%%%%%%%%%%%%%%%%%%%%%%%%%%%%%%%%%%%%%%%%%%%%%%%%%%%%%%%%%
\begin{abstract}

The uniqueness of the Perron vector of a nonnegative block matrix associated to a 
multiplex network is discussed. The conclusions come from the relationships between the 
irreducibility of some nonnegative block matrix associated to a multiplex network and 
the irreducibility of the corresponding matrices to each layer as well as the 
irreducibility of the adjacency matrix of the projection network. In addition the 
computation of that Perron vector in terms of the Perron vectors of the blocks is also 
addressed. Finally we present the precise relations that allow to express the Perron 
eigenvector of the multiplex network in terms of the Perron eigenvectors of its layers.
\end{abstract}
%%%%%%%%%%%%%%%%%%%%%%%%%%%%%%%%%%%%%%%%%%%%%%%%%%%%%%%%%%%%%%%%%%%%%

%%%%%%%%%%%%%%%%%%%%%%%%%%%%%%%%%%%%%%%%%%%%%%%%%%%%%%%%%%%%%%%%%%%%%
\section{Introduction and notation}\label{sec:intro}
%%%%%%%%%%%%%%%%%%%%%%%%%%%%%%%%%%%%%%%%%%%%%%%%%%%%%%%%%%%%%%%%%%%%%
% Aqui van los contenidos de la seccion de introduccion
% NO HACE FALTA PONER NINGUNA CABECERA NI NADA

%%%%%%%%%%%%%%%%%%%%%%%%%%%%%%%%%%%%%%%%%%%%%%%%%%%%%%%%%%%%%%%%%%%%%

%%%%%%%%%%%%%%%%%%%%%%%%%%%%%%%%%%%%%%%%%%%%%%%%%%%%%%%%%%%%%%%%%%%%%

Very recently some relevant aspects in the theory of multiplex networks have been considered with the help of an adequate matrix or tensor representation of the networks, particularly some related to the analysis of the eigenvalues and the eigenvectors of a matrix \cite{DSCKMPGA,DSGA13,DSGA14,Sanco14,SRCFGB13,Sodo13}.

This analysis typically includes the study of the existence and uniqueness of a positive and normalized eigenvector (Perron vector), whose existence is guaranteed if the corresponding matrix is irreducible (by using the classical Perron-Frobenius theorem). As for the spectral properties, it is possible to relate the irreducibility of such a matrix with the irreducibility in each layer and the irreducibility in the corresponding matrix of the projection network \cite{DSGA13,SRCFGB13}.
 
Some of these considerations are properly addressed with the help of the Perron vector of the  block matrix which represents the multiplex structure.
%%%%%%%%%%%%%%%%%%%%

The main goal of this paper is twofold. Firstly we show the uniqueness of the Perron eigenvector of the nonnegative block matrix associated to a multiplex network when the matrices of the layers and the matrix of connections between layers (or influence matrix) have some properties. Secondly we show how the Perron vector of the multiplex network relates to the lower-dimension Perron vectors of the layers and the Perron vector of the influence matrix in a precise way. Remarkably this  relationship is shown to be non linear; thus it becomes evident that the information framed in a multiplex network goes beyond a simple linear combination of the information provided by the layers.

The paper is divided in four sections. The first and second sections contain the notation employed and some background as well as a detailed description of the matrix products used along. The third section is entirely devoted to justifying the existence and uniqueness of the Perron eigenvector of the multiplex structure while the fourth section presents the precise (non linear) relations that allow to express  the Perron eigenvector of a multiplex network in terms of the Perron eigenvectors of its layers. The computations of this section are collected in a final appendix.

In the rest of the paper a {\sl multiplex network} is a set $\mathcal{M}=\{S_1, \cdots, S_m \}$   ($m\in \N$) of (directed or undirected, weighted or unweighed) complex networks $S_\ell=(X,E_\ell)$ (each of them called a {\sl layer} or {\sl state} of the multiplex network) that share the set of nodes $X=\{1,\cdots,n\}$. The adjacency matrix of each layer $S_\ell$ will be denoted by $A_\ell=(a_{ij}(\ell))\in \mathbb{R}^{n\times n}$.

In many situations, if we consider a multiplex network $\mathcal{M}$ of $m\in\N$ layers $\{S_1, \cdots, S_m \}$, we also take a {\sl influence matrix} $0\le W=(w_{ij})\in \mathbb{R}^{m\times m}$, where $w_{ij}$ measures the {\sl influence} of the layer $S_i$ in the layer $S_j$. Note that if we consider a random walker in a multiplex network, then each $w_{ij}$ can be understood as the probability of the  walker jumping from layer $S_i$ to layer $S_j$ (i.e. $W$ is the {\sl transition matrix} between the states of the multiplex network in the stochastic process given by a multiplex random walker) and therefore $W$ is a row stochastic matrix. Hence in the rest of the paper, we will always assume that the influence matrices $W$ are row stochastic.

Given a multiplex network $\mathcal{M}$ several (monoplex) networks that give valuable information about $\mathcal{M}$ can be associated to it. A first example of these (monoplex) networks is the unweighted {\sl projection network} $proj(\mathcal{M})=(X,E)$, where $X$ is the same set of nodes of the layers of $\mathcal{M}$ and
\[
E= \left(\bigcup_{\ell=1}^m E_\ell\right).
\]
It is clear that if $A=(a_{ij})\in \mathbb{R}^{n\times n}$ is the adjacency matrix of $proy(\mathcal{M})$, then
\[
a_{ij} =\left\{
\begin{tabular}{ll}
1  & \text{if $a_{ij}(\ell)=1$ for some $1\le\ell\le m$} \\
0  & \text{otherwise.}
\end{tabular}
\right.
\]

A first approach to the concept of multiplex networks could suggest that these new objects are actually (monolayer) networks with some (modular) structure in the mesoscale. It is clear that a (monolayer) network $\tilde{\mathcal{M}}$ can be associated to $\mathcal{M}$ as follows: $\tilde{\mathcal{M}}=(\tilde X,\tilde E)$, where $\tilde X$ is the disjoint union of all the nodes of $S_1,\cdots,S_m$, i.e.
\[
\tilde X=\bigcup_{1\le i\le m} X_i=\left\{(i,k)|\,\,i=1,\dots,n,\,\,k=1,\dots,m\right\}
\]
and $\tilde E$ is given by
\[
\tilde E=\left\{\left((i,k),(j,k)\right)|\,\, (i,j)\in E_k,\, 1\le k\le m\right\}
\bigcup\{\left((i,k),(i,l)\right) |\,\, i\in X,\, 1\le k\ne l\le m\}.
\]
Note that $\tilde{\mathcal{M}}$ is a (monolayer) network with $n\cdot m$ nodes whose adjacency matrix can be written as the block matrix
\[
\tilde A=\left(
\begin{array}{c|c|c|c}
A_1    & I_n & \cdots & I_n \\ \hline
I_n & A_2    & \cdots & I_n \\ \hline
\vdots & \vdots & \ddots & \vdots \\ \hline
I_n & I_n & \cdots & A_m
\end{array}
\right)\in \mathbb{R}^{nm\times nm}.
\]

It is important to remark that the behaviours of $\mathcal{M}$ and $\tilde{\mathcal{M}}$ are related but they are different since a single node of $\mathcal{M}$ belonging to several layers corresponds to $m$ different nodes in $\tilde{\mathcal{M}}$. Hence the properties and behaviours of corresponding (monolayer) network $\tilde{\mathcal{M}}$ could be understood as a kind of non-linear {\sl quotient} of the properties of the a multilayer $\mathcal{M}$.

Other examples of (monoplex) networks associated to a multiplex network $\mathcal{M}$ that give valuable information about the properties of $\mathcal{M}$ come from the study of several structural and dynamical properties of $\mathcal{M}$. In this paper we will consider the associated monoplex networks coming from the study of the eigenvector centrality of multiplex networks \cite{SRCFGB13} and from random walkers in multiplex networks \cite{DSGA13}.

If we want to extend the concept of eigenvector centrality to multiplex network, in \cite{SRCFGB13} the concept of global heterogeneous centrality of a multiplex network $\mathcal{M}$ with influence matrix $W$ is introduced from the Perron vector of the block matrix
\[
\B_0=
\left(
\begin{array}{c|c|c|c}
w_{11}A_1^t    & w_{21}A_2^t & \cdots & w_{m1}A_m^t \\ \hline
w_{12}A_1^t & w_{22}A_2^t    & \cdots & w_{m2}A_m^t \\ \hline
\vdots & \vdots & \ddots & \vdots \\ \hline
w_{1m}A_1^t & w_{2m}A_2^t & \cdots & w_{mm}A_m^t
\end{array}
\right)\in \mathbb{R}^{nm\times nm},
\]
where $A_\ell^t$ is the transpose of the adjacency matrix of layer $S_\ell$. Note that this kind of block matrix also appears if we consider some random walkers in multiplex networks. In this case, the distribution of the stationary state of the random walker is given from the Perron vector of the block matrix
\[
\B_1=
\left(
\begin{array}{c|c|c|c}
w_{11}L_1^t    & w_{21}L_2^t & \cdots & w_{m1}L_m^t \\ \hline
w_{12}L_1^t & w_{22}L_2^t    & \cdots & w_{m2}L_m^t \\ \hline
\vdots & \vdots & \ddots & \vdots \\ \hline
w_{1m}L_1^t & w_{2m}L_2 ^t & \cdots & w_{mm}L_m^t
\end{array}
\right)\in \mathbb{R}^{nm\times nm},
\]
where $L_\ell^t$ is the transpose of the row normalization of the adjacency matrix of layer $S_\ell$, i.e. if $L_\ell=(L_{ij}(\ell))$, then for each $1\le i,j\le n$
\[
L_{ij}(\ell)=\frac {a_{ij}(\ell)}{\displaystyle\sum_k a_{ik}(\ell)}.
\]
Note that each $L_\ell$ is row stochastic and therefore $L_\ell^t$ is colum stochastic.

Similarly, in \cite{DSGA13} a general framework for random walkers in multiplex networks is introduced and the  distribution of the stationary states of these random walkers are given from the Perron vector of some block matrices. In particular, if we consider random walkers with no cost in the transition between states, the distribution of the stationary state is given in terms of the  Perron vector of
\[
\B_2=
\left(
\begin{array}{c|c|c|c}
w_{11}L_1^t    & w_{21}L_1^t & \cdots & w_{m1}L_1^t \\ \hline
w_{12}L_2^t & w_{22}L_2^t    & \cdots & w_{m2}L_2^t \\ \hline
\vdots & \vdots & \ddots & \vdots \\ \hline
w_{1m}L_m^t & w_{2m}L_m ^t & \cdots & w_{mm}L_m^t
\end{array}
\right)\in \mathbb{R}^{nm\times nm},
\]
while if
we consider random walkers with cost in the transition between states, the distribution of the stationary state is given in terms of the  Perron vector of
\[
\B_3=
\left(
\begin{array}{c|c|c|c}
w_{11}L_1^t    & w_{21}I_n & \cdots & w_{m1}I_n \\ \hline
w_{12}I_n & w_{22}L_2^t    & \cdots & w_{m2}I_n \\ \hline
\vdots & \vdots & \ddots & \vdots \\ \hline
w_{1m}I_n & w_{2m}I_n & \cdots & w_{mm}L_m^t
\end{array}
\right)\in \mathbb{R}^{nm\times nm}.
\]

As we will see in section 3, it can be proven that, under some hypotheses, if the adjacency matrix of the projection network is irreducible, then these matrices are also irreducible and hence the corresponding random walkers have a unique stationary state.

%In \cite{DSGA14} the $nm\times nm$ adjacency matrix, $\mathcal{A}$,
%of the  multiplex is given by
%
%$$\mathcal{A} = \mathsf{A} + \mathsf{C}$$
%where $\mathsf{A}$ the direct
%sum of the adjacency matrices of the $m$ layers:
%
%$$
%\mathsf{A}=\left(
%\begin{array}{c|c|c|c}
%A_{1}    & 0 & \cdots & 0 \\ \hline 0    & A_{2} & \cdots & 0\\
%\hline \vdots & \vdots & \ddots & \vdots \\ \hline 0    & 0 & \cdots
%& A_{m}
%\end{array}
%\right) \, ,
%$$
%and $\mathsf{C}$ contains the
%interlayer interactions:
%
%$$
%\mathsf{C}=\left(
%\begin{array}{c|c|c|c}
%C_{11}    & C_{12} & \cdots & C_{1m} \\ \hline C_{21}    & C_{22} &
%\cdots & C_{2m}\\ \hline \vdots & \vdots & \ddots & \vdots \\ \hline
%C_{m1}    & C_{m2} & \cdots & C_{mm}
%\end{array}
%\right) \, ,
%$$
%where $C_{\alpha \beta}$ is the $n \times n$-matrix of interactions between nodes in layer $\alpha$ and nodes in layer
%$\beta$.
%
%\

%In \cite{EG14} the authors consider $C_{\alpha \beta} = C_{\beta \alpha} = C =
%\omega I$, for all layers $\alpha \neq \beta$, $C_{\alpha \alpha} =
%0$, where $\omega$ is a parameter that gives information of the strength of the
%interlayer interactions and $I$ is the $n \times n$
%identity matrix. Then, we can write
%
%$$\mathsf{C} = (E - I) \otimes C \, ,$$
%where $E$ is an all-ones $m \times m$ matrix and $I$ the $m \times
%m$ identity matrix.

\
This kind of arguments can be also applied to the supra-Laplacian $\mathcal{L}$ of a  multiplex (\cite{Gomez13} and \cite{Sodo13}) since we have the splitting
$$\mathcal{L} = \mathcal{L}^m + \mathcal{L}^I \, ,$$
%where $\mathcal{L}^m$ stands for the supra-Laplacian of the
%independent layers and $\mathcal{L}^I$ for the interlayer
%supra-Laplacian.
% the intralayer connectivity of
%layer $\alpha$ is expressed as an adjacency or strengths matrix
%$A_{\alpha}$ whose corresponding Laplacian is $L_{\alpha} =
%D_{\alpha} - A_{\alpha}$, where $D_{\alpha}$ is the diagonal matrix
%of the nodes intralayer strengths. There the
%interlayer connectivity is supposed identical for all nodes, thus the interlayer network matrix  $A_I$ can be defined as the ${m \times m}$-matrix whose 
%components represent the strength of the connection between every
%pair of layers, and the associated interlayer Laplacian is $L_I =
%D_I - A_I$.
%
%\
where $\mathcal{L}^m$ stands for the supra-Laplacian of the
independent layers and $\mathcal{L}^I$ for the interlayer
supra-Laplacian. The first one is just the direct sum of the
intralayer Laplacians,
$$
\mathcal{L}^L =\left(
\begin{array}{c|c|c|c}
L_{1}    & 0 & \cdots & 0 \\ \hline 0    & L_{2} & \cdots & 0\\
\hline \vdots & \vdots & \ddots & \vdots \\ \hline 0    & 0 & \cdots
& L_{m}
\end{array}
\right) \, ,
$$
while the interlayer supra-Laplacian may be expressed as the
Kronecker (or tensorial) product (see section  \ref{sec:definitions}) of the interlayer Laplacian and the
$n \times n$ identity matrix $I$,
$$\mathcal{L}^I = L^I \otimes I \, .$$

% Final del documento de introduccion
%\input{intror}

%%%%%%%%%%%%%%%%%%%%%%%%%%%%%%%%%%%%%%%%%%%%%%%%%%%%%%%%%%%%%%%%%%%%%
 \section{Block Hadamard and Block Khatri-Rao Products}
 \label{sec:definitions}
%%%%%%%%%%%%%%%%%%%%%%%%%%%%%%%%%%%%%%%%%%%%%%%%%%%%%%%%%%%%%%%%%%%%%

%\input{block}
%\input{block1}
% Aqui van los contenidos de la seccion de definicion de los productos por bloques
% NO HACE FALTA PONER NINGUNA CABECERA NI NADA

In addition to the conventional matrix product, there are some other
matrix products which will be used throughout this paper.

\

Note that, for example, 

\

\[
\B_1=
\left(
\begin{array}{c|c|c}
w_{11}L_1^t    & \cdots & w_{m1}L_m^t \\ \hline
\vdots & \ddots & \vdots \\ \hline
w_{1m}L_1^t & \cdots & w_{mm}L_m^t
\end{array}
\right),
\]
is the Hadamard product of

\[
\left(
\begin{array}{c|c|c}
w_{11}1_n    & \cdots & w_{m1}1_n \\ \hline
\vdots & \ddots & \vdots \\ \hline
w_{1m}1_n & \cdots & w_{mm}1_n
\end{array}
\right)
\, \, \, \mbox{and} \, \, \,
\left(
\begin{array}{c|c|c}
L_1^t    & \cdots & L_m^t \\ \hline
\vdots & \ddots & \vdots \\ \hline
L_1^t & \cdots & L_m^t
\end{array}
\right),
\]
where $1_n$ the matrix $n\times n$ whose components are all equal to one, or the generalized Khatri-Rao product of

\[
\left(
\begin{array}{c|c|c}
w_{11}    & \cdots & w_{m1} \\ \hline
\vdots & \ddots & \vdots \\ \hline
w_{1m} & \cdots & w_{mm}
\end{array}
\right)
\, \, \, \mbox{and} \, \, \,
\left(
\begin{array}{c|c|c}
L_1^t    & \cdots & L_m^t \\ \hline
\vdots & \ddots & \vdots \\ \hline
L_1^t & \cdots & L_m^t
\end{array}
\right).
\]

\

This section provides a brief survey on such definitions and basic
properties without proofs. Throughout this section we refer to some standard 
references of matrix theory for details.

\

Let us consider two matrices $A$ and $B$ of  $m \times n$ and $p \times
q$ orders respectively. Let us suppose that $A = [A_{ij}]$ is partitioned with $A_{ij}$ of
order $m_i \times n_j$ ($A_{ij}$ is the $(i,j)^{th}$ block submatrix of $A$) and let $B =
[B_{kl}]$ be partitioned with $B_{kl}$ of order $p_k \times q_l$ ($B_{kl}$ is
the $(k,l)^{th}$ block submatrix of $B$). Denote by $m = \sum^t_{i=1} m_i, \, n =
\sum^d_{j=1} n_j ,  p = \sum^{u}_{k=1} p_k,$ and $q = \sum^v_{l=1}
q_l$). For simplicity, we say that $A$ and $B$ are
\textit{compatible partitioned} if $A = [A_{ij}]^t_{i,j=1}$ and $B =
[B_{ij}]^t_{i,j=1}$ are square matrices of order $m \times m$ and
partitioned, respectively, with $A_{ij}$ and $B_{ij}$ of order $m_i
\times m_j$ ($m = \sum^t_{i=1} m_i = \sum^t_{j=1} m_j$).

\

Let $A \otimes B$, $A \circ B$, $A \Theta B$, and $A * B$ be the
Kronecker, Hadamard, Tracy-Singh, and Khatri-Rao products,
respectively, of $A$ and $B$. All the definitions of the mentioned four
matrix products can be found in \cite{Liu99}, \cite{Liu02} as follows:

\

(i) \textit{Kronecker product}

\

The Kronecker product of matrices is also called the tensor product,
or direct product of matrices. This product is applicable to any two
matrices. We refer to \cite{HJ91} for a complete discussion.

\

Let $A = (a_{ij}) \in \mathbb{R}^{m \times n}$ and $B = (b_{ij}) \in \mathbb{R}^{p \times q}$. The Kronecker product of $A$ and $B$ is defined as

\[
A \otimes B = (a_{ij}B)_{ij} = \left(
\begin{array}{cccc}
a_{11}B & a_{12}B & \cdots & a_{1n}B \\
a_{21}B & a_{22}B & \cdots & a_{2n}B \\
\vdots & \vdots & \ddots & \vdots \\
a_{m1}B & a_{m2}B & \cdots & a_{mn}B
\end{array}
\right) \in \mathbb{R}^{mp\times nq}.
\]

(ii) \textit{Hadamard product}

\

The Hadamard product (elementwise multiplication), also referred to as
the \textit{Schur product}, arises in a wide variety of mathematical
applications such as covariance matrices for independent zero mean
random vectors and characteristic functions in probability theory.
The reader is referred to \cite{HJ91}, \cite{Zhang04},
\cite{Schott05} for more details about it.

\

Let $A = (a_{ij}), B = (b_{ij}) \in \mathbb{R}^{m \times n}$. The Hadamard
product of $A$ and $B$ is defined as

\[
A \circ B = (a_{ij}b_{ij})_{ij} = \left(
\begin{array}{cccc}
a_{11}b_{11} & a_{12}b_{12} & \cdots & a_{1n}b_{1n} \\
a_{21}b_{21} & a_{22}b_{22} & \cdots & a_{2n}b_{2n} \\
\vdots & \vdots & \ddots & \vdots \\
a_{m1}b_{m1} & a_{m2}b_{m2} & \cdots & a_{mn}b_{mn}
\end{array}
\right) \in \mathbb{R}^{m \times n}.
\]

\

(iii) \textit{Tracy-Singh product}
$$A \Theta B = [A_{ij} \Theta B]_{ij} = [ [ A_{ij} \otimes B_{kl}]_{kl} ]_{ij} ,$$
where $A = [A_{ij}]$, $B = [B_{kl}]$ are partitioned matrices of
order $m \times n$ and $p \times q$, respectively, $A_{ij}$ is of
order $m_i \times n_j$ , $B_{kl}$ of order $p_k \times q_l$ ,
$A_{ij} \otimes B_{kl}$ of order $m_ip_k \times n_jq_l$, $A_{ij}
\Theta B$ of order $m_ip \times n_jq$ ($m = \sum^t_{i=1} m_i, \, n =
\sum^d_{j=1} n_j, \, p = \sum^{u}_{k=1} p_k, \, q = \sum^v_{l=1}
q_l$), and $A \Theta B$ of order $mp \times nq$;

\

In order to avoid confusion we use parentheses for ordinary
matrices, whose entries are numbers, multiplied as usual, and square
brackets for cores (core matrices), whose entries are blocks.

(iv) \textit{Generalized Khatri-Rao product}
$$A \ast B = [A_{ij} \otimes B_{ij}]_{ij} \,$$
where $A = [A_{ij}]$, $B = [B_{ij}]$ are partitioned matrices of
order $m \times n$ and $p \times q$, respectively, $A_{ij}$ is of
order $m_i \times n_j$ , $B_{kl}$ of order $p_i \times q_j$ ,
$A_{ij} \otimes B_{ij}$ of order $m_ip_i \times n_jq_j$ ($m =
\sum^t_{i=1} m_i, \, n = \sum^d_{j=1} n_j, \, p = \sum^t_{i=1} p_i,
\, q = \sum^d_{j=1} q_j$ ), and $A \ast B$ of order $M{\times}N$ ($M =
\sum^t_{i=1} m_ip_i, \, N = \sum^d_{j=1} n_jq_j$).

\

Note that the generalized Khatri-Rao product is defined based on a
particular matrix partitioning, i.e., different matrix partitionings
will lead to different results. Note also that the Kronecker
product, the Hadamard product and the Khatri-Rao product
\cite{KhatriRao68}, \cite{Rao70} are all special cases of the
generalized Khatri-Rao product based on different matrix
partitionings.

\

Recall that given two matrices $A$ and $B$ with the same number of
columns, $m$, and denoting their columns by $a_i$ and $b_i$,
respectively, the (column-wise) \em Khatri-Rao \em product is defined as $A
\ast B = [a_1 \otimes b_1, a_2 \otimes b_2, \cdots, a_m \otimes
b_m]$ (we refer to \cite{LS82}, \cite{Zhang04} or \cite{LT08} for
details). Note that the Khatri-Rao product can be constructed by
selecting columns from the Kronecker product. To show this, define
the Kronecker selection matrix $S_m = I_m \ast I_m$ and verify $A
\ast B = (A \otimes B)S_m$, where $I_m$ is the identity matrix in
$\mathbb{R}^{m \times m}$.

\

Additionally, \cite{Liu99} shows that the generalized Khatri-Rao
product can be viewed as a generalized Hadamard product and the
Tracy-Singh product as a generalized Kronecker product, as follows:

\

(1) for a nonpartitioned matrix $A$, their $A \Theta B$ is $A
\otimes B$;

\

(2) for nonpartitioned matrices $A$ and $B$ of order $m \times n$,
their $A \ast B$ is $A \circ B$.

\

The Khatri-Rao and Tracy-Singh products are related by the following
relation \cite{Liu99}, \cite{Liu02} :
$$A \ast B = Z^T_1 (A \Theta B)Z_2,$$
where $A = [A_{ij}]$ is partitioned with $A_{ij}$ of order $m_i
\times n_j$ and $B = [B_{kl}]$ is partitioned with $B_{kl}$ of order
$p_k \times q_l$ ($m = \sum^t_{i=1} m_i, \, n = \sum^d_{j=1} n_j, \,
p = \sum^u_{k=1} p_k, \, q = \sum^v_{l=1} q_l)$, $Z_1$ is an $mp
\times r$ ($r = \sum^t_{i=1} m_ip_i$) matrix of zeros and ones, and
$Z_2$ is an $nq \times s$ ($s = \sum^d_{j=1} n_jq_j$) matrix of
zeros and ones such that $Z^T_1 Z_1 = I_r, Z^T_2 Z_2 = I_s$ ($I_r$
and $I_s$ are $r \times r$ and $s \times s$ identity matrices,
resp.).

\

In particular, if $m = n$ and $p = q$, then there exists a $mp
\times r$ ($r = \sum^t_{i=1} m_ip_i$) matrix $Z$ such that $Z^TZ =
I_r$ ($I_r$ is an $r \times r$ identity matrix) and $A \ast B =
Z^T(A \Theta B)Z$. Here
\[
Z= \left[
\begin{array}{ccc}
Z_1 &  & \\
 & \ddots & \\
 &  & Z_t
\end{array}
\right] \, ,
\]
where each $Z_i = [ 0_{i1}  \, \, \cdots \, \, 0_{ii-1} \, \, I_{m_i
p_i} \, \, 0_{ii+1} \, \, \cdots \, \, 0_{it}]^T$ is a real matrix
of zeros and ones, and $0_{ik}$ is a $m_ip_i \times m_ip_k$ zero
matrix for any $k \neq i$. Note also that $Z^T_i Z_i = I$ and
$$Z^T_i (A_{ij} \Theta B)Z_j = Z^T_i (A_{ij} \otimes B_{kl})_{kl}Z_j = A_{ij} \otimes B_{ij}, \, \, i,j = 1,2,\cdots,t.
$$

\

The generalized Khatri-Rao product was also used, e.g., in
\cite{XuStoicaLi06}.

\

Let $A$ and $B$ be matrices respectively expressed as $r \times t$
and $t \times u$ block matrices
\[
A = \left(
\begin{array}{cccc}
A_{11} & A_{12} & \cdots & A_{1t} \\
A_{21} & A_{22} & \cdots & A_{2t} \\
\vdots & \vdots & \ddots & \vdots \\
A_{r1} & A_{r2} & \cdots & A_{rt}
\end{array}
\right) \, \, \ \text{and}\ \, \, B = \left(
\begin{array}{cccc}
B_{11} & B_{12} & \cdots & B_{1u} \\
B_{21} & B_{22} & \cdots & B_{2u} \\
\vdots & \vdots & \ddots & \vdots \\
B_{t1} & B_{t2} & \cdots & B_{tu}
\end{array}
\right) ,
\]
where  each $A_{ij}$ ($i = 1, 2, \cdots, r$ and $j = 1, 2,
\cdots, t$) is an $m \times p$ matrix, and each $B_{ij}$ ($i=1,2,
\cdots,t$ and $j=1, 2, \cdots, u$) is a $n \times q$ matrix. In
\cite{SZ93} the \em strong Kronecker \em product is defined for two matrices $A$ and $B$ of dimensions $r \times t$ and $t \times u$ respectively as the matrix:

\[
C = \left(
\begin{array}{cccc}
C_{11} & C_{12} & \cdots & C_{1u} \\
C_{21} & C_{22} & \cdots & C_{2u} \\
\vdots & \vdots & \ddots & \vdots \\
C_{r1} & C_{r2} & \cdots & C_{ru}
\end{array}
\right) ,
\]
where each
$$C_{ij} = A_{i1} \otimes B_{1j} + A_{i2} \otimes B_{2j} + \cdots + A_{it} \otimes B_{tj} \, ,$$
is an $mn \times pq$ matrix. It is important to note that the
operation is fully determined only after the parameters $r$, $t$,
and $u$ are fixed. Generally, the partitioning of the matrices will be
clear from the context, and then we call $C$ the strong Kronecker
product of $A$ and $B$, denoted by $A \circledast B$. The strong
Kronecker product, developed in \cite{SZ93}, supportes the
analysis of certain orthogonal matrix multiplication problems. The
strong Kronecker product is considered a powerful matrix
multiplication tool for Hadamard and other orthogonal matrices from
combinatorial theory \cite{LN94}. In \cite{Pitsianis98} the
strong Kronecker product is shown to be a matrix multiplication in a
permuted space. Similarly, if $m = n$ and $p = q$, the \em strong
Hadamard \em product $A \odot B$ of $A$ and $B$ is defined in
\cite{CT08} as

\[
A \odot B = \left(
\begin{array}{cccc}
D_{11} & D_{12} & \cdots & D_{1u} \\
D_{21} & D_{22} & \cdots & D_{2u} \\
\vdots & \vdots & \ddots & \vdots \\
D_{r1} & D_{r2} & \cdots & D_{ru}
\end{array}
\right) ,
\]
where each
$$D_{ij} = A_{i1} \circ B_{1j} + A_{i2} \circ B_{2j} + \cdots + A_{it} \circ B_{tj} \, ,$$
is an $m \times p$ matrix.

\

Let $A = (A_{ij})$ and $B = (B_{ij})$ be $p \times p$ block matrices
in which each block is an $n \times n$ matrix. In \cite{HMN91} 
a \em block Hadamard \em product $ A \square B$ is defined by $A \square B :=
(A_{ij}B_{ij})$, where $A_{ij}B_{ij}$ denotes the usual matrix
product of $A_{ij}$ and $B_{ij}$.

\

There are other definitions of partitioned matrix products, see for instance  \cite{GK2012} where a generalized Kronecker product for block matrices is defined.

% Final del documento de definicion de los productos por bloques

%%%%%%%%%%%%%%%%%%%%%%%%%%%%%%%%%%%%%%%%%%%%%%%%%%%%%%%%%%%%%%%%%%%%%
 \section{Irreducibility and uniqueness of Block Perron Vectors through properties of the blocks}
 \label{sec:uniquenessbis}
%  \label{sec:uniqueness}
%%%%%%%%%%%%%%%%%%%%%%%%%%%%%%%%%%%%%%%%%%%%%%%%%%%%%%%%%%%%%%%%%%%%%

%{\color{blue}
%En esta seccion trataremos los siguientes temas:
%\begin{enumerate}
% \item Veremos como garantizar el que la matriz por bloques sea irreducible en terminos de la irreducibilidad de la proyeccion y otras propiedades de la matriz de interconexion $W$.
% \item Recordaremos nuestro resultado visto para la multiplex eigenvector centrality y lo extenderemos a los otros casos.
% \item En el caso del paseante aleatorio con coste de cambio de capa daremos el teorema para $W$ irreducible.
% \item Veremos un contraejemplo en el que se muestra que la irreducibilidad de la proyeccion y de la matriz de interconexion no es suficiente.
%\end{enumerate}
%De esta seccion se puede encargar principalmente Miguel y Luis}.

%
%{\color{red}
%Los contenidos de esta seccion se incluiran en el archivo {\bf uniqueness.tex} en lugar de incluirlos en el archivo Perron\_Multiplex-002.tex, para evitar que tengamos versiones diferentes.
%}

%\input{uniquenessbis}

In this section we will discuss irreducibility of the block matrices that appear in our different descriptions of multiplex networks. 
Let us start by introducing some notation.

\subsection{Products of block matrices}\label{ssec:products}

In the sequel we will consider block matrices consisting of $m^2$ blocks of dimensions $n\times n$ with real nonnegative coefficients:
$$
P=\left(
\begin{array}{c|c|c|c}
P_{11}    & P_{12} & \cdots & P_{1m} \\ \hline
P_{21}    & P_{22} & \cdots & P_{2m}\\ \hline
\vdots & \vdots & \ddots & \vdots \\ \hline
P_{m1}    & P_{m2} & \cdots & P_{mm}
\end{array}
\right)%\in M_{nm}^+
,\quad P_{ij}\in \mathbb{R}^{n\times n}.
$$
The set of all this matrices will be denoted by  $M_{nm,n}^+(\mathbb{R})$, or simply $M_{nm,n}^+$.

For two such  block matrices $P$ and $P'$, %besides the usual product, that we denote by $P\cdot P'$, we will consider
%\marginpar{\tiny Luis: cambiar nombre, o s\'{\i}mbolo, seg\'un literatura} 
let us consider the strong Hadamard product defined above:
$$
(P \odot P')_{ij}=\sum_{k=1}^mP_{ik}\circ P'_{kj},
$$
where $P_{ik}\circ P'_{kj}$ denotes the Hadamard product (i.e. the componentwise product) of the blocks $P_{ik}$ and $P'_{kj}$.

For a given a sequence of $n\times n$ matrices $(A_1,\dots,A_m)$ we can consider the diagonal block matrix $\underline{A}$ matrix defined by:
$$
\underline{A}=\left(
\begin{array}{c|c|c|c}
A_{1}    & 0 & \cdots & 0 \\ \hline
0    & A_{2} & \cdots & 0\\ \hline
\vdots & \vdots & \ddots & \vdots \\ \hline
0    & 0 & \cdots & A_{m}
\end{array}
\right)
$$
We will denote by $I_n$ the $n\times n$ identity matrix, and by $1_n$ the matrix $n\times n$ whose components are all equal to one. Then the identity element of the product $\odot$ is $\underline{1}$, that is, the diagonal block matrix given by the sequence $(1_n,\dots,1_n)$.

Let us denote by $R_2$ the Boolean algebra with two elements $\{0,1\}$, on which we have two operations, namely:

\begin{table}[h]
\centering
\begin{tabular}{|c|c|c|}
\hline
$+$&$0$&$1$\\\hline
$0$&$0$&$1$\\\hline
$1$&$1$&$1$\\\hline
\end{tabular}\quad\quad\quad
\begin{tabular}{|c|c|c|}
\hline
$\,\cdot\,$&$0$&$1$\\\hline
$0$&$0$&$0$\\\hline
$1$&$0$&$1$\\\hline
\end{tabular}
\end{table}
Then, for every nonnegative matrix $P\in M_{nm,n}^+$ we may define its {\it booleanization} $\beta(P)$ as the $nm$ block matrix with coefficients in $R_2$ given by:
$$
(\beta(P)_{ij})_{kr}=\left\{\begin{array}{ll}\vspace{0.2cm}1&\mbox{si }(P_{ij})_{kr}\neq 0\\0&\mbox{si }(P_{ij})_{kr}= 0 \end{array} \right. 
$$ 
for all $i,j=1,\dots,m$, $k,r=1,\dots,n$.

Notice that the map $\beta: M_{nm,n}^+\to M_{nm,n}(R_2)$ preserves, by definition, sums, and the usual, Hadamard and strong Hadamard products; notice also that the irreduciblity of a nonnegative matrix, which is the main topic of this section, depends only on its booleanization, which can be thought of as a matrix-representation of the graph defined by the matrix.

A partial order can be defined in $M_{nm,n}(R_2)$ as $B\leq B'$ if and only if there exists $B''\in M_{nm,n}(R_2)$ such that $B+B''=B'$. It becomes  obvious that, if $P\in M_{nm,n}^+$ is irreducible, then any other matrix $P' \in M_{nm,n}^+$ satisfying $\beta(P)\leq\beta(P')$ must be irreducible as well.

Finally we note that for every block matrix $P\in M_{nm,n}^+$ a new block matrix $\widehat{P}\in M_{nm,m}^+$ can be defined by reordering the coefficients as follows:
$$
(\widehat{P}_{kr})_{ij}=(P_{ij})_{kr},\quad i,j=1,\dots,n,\quad k,r=1,\dots,m
$$
This new matrix is formed by $n^2$ blocks of dimension $m\times m$.
%, and we will call it {\color{blue} name?}.

\subsection{Block matrices for multiplex networks}\label{ssec:blocks}

In order to model multiplex networks as they appear in nature, scientists have introduced several types of special block matrices. Generally speaking, they are all constructed upon the following data:
\begin{itemize}
\item A set of $m$ nonnegative $n\times n$ matrices $\{A_1,\dots,A_m\}$, each $A_i$ is the the adjacency matrix of the $i$-layer belonging to the multiplex network. In this context 
the matrix  $\overline{A}:=\frac{1}{m}\sum_{i=1}^mA_i$, whose associated graph is the projection network of the complex network under study, is considered. %Dividing by $L$ allows us to say that, if the $A_i$'s are row-stochastic, then $\overline{A}$ is row-stochastic as well. Obviously, this does not affect the irreducibility of $\overline{A}$.
\item Two $nm\times nm$ nonnegative block matrices, encoding the interrelation between layers:
$$
W=\left(
\begin{array}{c|c|c|c}
W_{11}    & W_{12} & \cdots & W_{1m} \\ \hline
W_{21}    & W_{22} & \cdots & W_{2m}\\ \hline
\vdots & \vdots & \ddots & \vdots \\ \hline
W_{m1}    & W_{m2} & \cdots & W_{mm}
\end{array}
\right),\quad V=\left(
\begin{array}{c|c|c|c}
V_{11}    & V_{12} & \cdots & V_{1m} \\ \hline
V_{21}    & V_{22} & \cdots & V_{2m}\\ \hline
\vdots & \vdots & \ddots & \vdots \\ \hline
V_{m1}    & V_{m2} & \cdots & V_{mm}
\end{array}
\right)%\in M_{nm}^+.
$$
\end{itemize}
We may think of $W$ as the matrix encoding interrelations between layers (influence matrix), whereas $V$ represents interrelations between layers arising from the set of all the specific influences  that a node in a layer has over a node in another (not necessarily different) layer.

Then, upon this data, we consider the matrices:
$$
\B=\underline{A}\odot W+V\mbox{ and }\B'= W\odot\underline{A}+V,
$$
Notice that both $\B$ and $\B'$ have their own  eigenvector centrality. 

Two particular cases of the previous general scheme have a clear interest. 

\begin{itemize}
\item[1.] The term $V$ is identically zero. Then we have two block matrices $$ \B_1=\underline{A}\odot W\mbox{ and }\B'_1= W\odot\underline{A}$$ (this is the situation when modelling random walkers with no cost for the sate transition).
% \dots ({\color{blue}description})

\item[2.] The term $W$ is equal to $\underline{1}$, so that our two block matrices are equal: $$\B_2=\underline{A}\odot \underline{1}+V=\underline{A}+V=\underline{1}\odot\underline{A}+V=\B'_2.$$
    Typically in this case one would ask $V$ to satisfy the following property:
    $$
    (\star)\,\,\,V_{ij}\mbox{ diagonal, for all }i,j%\mbox{ and }V_{ii}=0\mbox{ for all }i.
    $$
    In other words, the property $(\star)$ is satisfied whenever $\widehat{V}=\underline{B}$ being $B=(B_1,\dots,B_n)$ a sequence of $m\times m$ nonnegative matrices.
    
     Each matrix $B_j$ represent the way in which one may switch between layers, while staying at node $j$  (this is the situation when modelling random walkers with no cost for the sate transition).
     %{\color{blue}description of the usage of this type of multiplex network})
\end{itemize}
In search of irreducibility conditions we will work on this general scheme; this is the  content of the next subsection.
\subsection{Irreducibility conditions}\label{ssec:irred}

As announced the rest of the section is devoted to describing irreducibility conditions of the matrices described above. 
Since we are going to discuss irreducibility through is graph-theoretical counterpart --strong connectedness-- we need to introduce first some notation. 

Given a multiplex network determined by one of the matrices $\B$ (or $\B'$) described above, we will write $i\stackrel{k}{\to}j$ when the node $i$ is linked to the node $j$ in layer $k$, i.e. when the coefficient $(A_k)_{ij}$ is different from zero. We will now consider a new monoplex network with nodes $\tilde{X}=\{(i,k)|\,\,i=1,\dots,n,\,\,k=1,\dots,m\}$ and write $(i,k)\to (j,\ell)$ when the coefficient in the position $ij$ of the block $k\ell$ of $\B$ (or $\B'$) is different from zero. In other words, we consider the weighted graph $(\tilde{X},\B)$ (or $(\tilde{X},\B')$) supported on the monolayer network $\tilde{\mathcal{M}}$.

In the case 1, we will start by analizing the case in which the projected network is strongly connected, that is, in which $\overline{A}$ is irreducible. Unfortunately, in this case, even if $W$ is positive, very simple examples show that $\B_1$ and $\B'_1$ are not necessarily irreducible. However we may state that there exists a unique Perron vector for them.  

\begin{theorem}\label{thm:irred1}
With the same notation as above, assume that $\overline{A}$ is irreducible and $W$ is positive. Then $\B_1$ and $\B'_1$ have a unique Perron vector.
\end{theorem}

\begin{proof} We will present the proof of the uniqueness for $\B_1$, being the proof for $\B'_1$ analogous.

Note that the matrix $\B_1$ may have rows completely equal to zero, preventing it from being irreducible. If $W$ is strictly positive, this happens precisely if there exists a sink in the graph of one of the layers. In order to deal with this situation, we consider a permutation matrix $P$ that reorders the rows of $\B_1$ so that all the rows equal to zero appear in the first positions. Then the product $P\cdot \B_1\cdot P^t$ takes the form:
$$
P\cdot \B_1\cdot P^t=\left(\begin{array}{ccc|ccc}
0&\cdots&0&0&\cdots&0\\
\vdots&\ddots&\vdots&\vdots&&\vdots\\
0&\cdots&0&0&\cdots&0\\\hline
\star&\cdots&\star&&&\\
\vdots&&\vdots&&\mbox{\Huge{R}}&\\
\star&\cdots&\star&&&
\end{array}\right)
$$
and it suffices to show that $R$ is an irreducible matrix, because in this case the algebraic multiplicities of the spectral radius of $\B_1$ as an eigenvalue of $\B_1$ equals its multiplicity as an eigenvalue for $R$, which is equal to one.

In order to check the irreducibility of $R$ note first that, by the positivity of $W$:
\begin{equation}
\big(i\stackrel{k}{\to}j\big)\,\,\iff\,\, \big((i,k)\to (j,k)\big)\iff\,\, \big((i,k)\to (j,\ell)\big)\mbox{ for all }\ell=1,\dots,m%\,\,\iff\,\,\big(s_{ir}\to s_{jk}\mbox{ in }E^\otimes \,\,\forall r\big),
\label{eq:horiz}
\end{equation}
Considering then the weighted subgraph of $(\tilde{X},\B_1)$ associated to $R$, and denoting by $\tilde{X}_{R}$ its set of nodes, that is:
$$
\tilde{X}_{R}=\left\{\right(i,k)|\,\,i\stackrel{k}{\to}j\mbox{ for some }j\},
$$ 
it suffices to show that $(\tilde{X}_{R},{R})$ is strongly connected.

Let then $(i,k),(i',k')$ be two nodes of this subgraph. Since $(i,k)\in\tilde{X}_{R}$, there exist $j_1\in\{1,\dots,n\}$ such that $$(i,k)\to (j_1,\ell)\mbox{ for all }\ell.$$ 
Moreover, by hypothesis on $\overline{A}$, we know that there exist two sequences of indices $(j_1,\dots, j_r=i')$, $j_p\in\{1,\dots,n\}$, and $(k_2,\dots,k_{r})$, $k_p\in\{1,\dots,m\}$, such that:
$$
j_1\stackrel{k_2}{\to}j_2\stackrel{k_3}{\to}\dots\stackrel{k_r}{\to}j_r=i',
$$
and so $(j_p,k_{p+1})\to (j_{p+1},\ell)$ for all $\ell$. 
Summing up, we have a sequence of edges linking $(i,k)$ to $(i',k')$:
$$
(i,k)\to (j_1,k_2)\to (j_2,k_3)\to\dots\to(j_{r-1},k_r)\to(i',k').
$$
\end{proof}

\begin{remark}
Note that, denoting by $1_{nm}\in M_{nm,n}^+$ the matrix whose coefficients are all ones, the proof holds for every nonnegative block matrix $W$ satisfying $\beta(W)\geq\beta(\underline{A}\odot1_{nm})$ (or $\beta(W)\geq \beta(1_{nm}\odot\underline{A})$, when we are dealing with $\B'_1$).
\end{remark}

The next corollary is an immediate consequence of the previous proof:

\begin{corollary}
With the same notation as above, assume that $\overline{A}$ is irreducible and that $W$ is strictly positive. Assume moreover that each layer $A_k$ of the network has no sinks (respectively, no sources). Then $\B_1$ (resp. $\B'_1$) is irreducible.
\end{corollary}

Let us consider now the case 2. Here we will infer the irreducibility of $\B_2=\B_2'$ from properties of $(A_1,\dots,A_m)$ and $(B_1,\dots, B_n)$.

\begin{proposition}
With the same notation as above, assume that one of the following properties holds:
\begin{itemize}
  \item[(i)] $\overline{A}$ and every $B_i$ are irreducible. 
  \item[(ii)] Every $A_k$  and $\overline{B}$ are irreducible.
\end{itemize}
Then $\B_2$ is irreducible.
\end{proposition}

\begin{proof}
As usual, we will discuss the proof in terms of the subjacent networks.
In the first case, given two pairs $(i,k),(i',k')\in \tilde{X}$, the irreducibility of $\overline{A}$ provides a sequences of edges:
$$
i=j_0\stackrel{k_1}{\to}j_1\stackrel{k_2}{\to}\dots\stackrel{k_r}{\to}j_r=i'.
$$
That is, we have links $$(i=j_0,k_1)\to(j_1,k_1),\quad (j_1,k_2)\to (j_2,k_2),\quad\dots,\quad (j_{r-1},k_r)\to(i'=j_{r},k_r).$$
Denote $k_0:=k$, $k_{r+1}:=k'$. Then, the irreducibility of the $B_i$'s provides sequences of edges joining $(j_{p},k_{p})$ with $(j_p,k_{p+1})$ for all $p=0,\dots, r$. Joining all these sequence conveniently, we have a sequence of edges joining $(i,k)$ and $(i',k')$. The irreducibility of $\B_2$ under the second set of hypotheses is analogous.
\end{proof}

\begin{remark}
As we may see in this Proposition, in this second setup, the links within layers and between layers play a symmetric role. In this way, every theorem about $\B_2$ written in terms of $A$ and $B$ will always have a symmetric counterpart. 
\end{remark}

 \section{Computation of Block Perron Vectors in terms of low-dimensional vectors}
 \label{sec:eigenvector}
%%%%%%%%%%%%%%%%%%%%%%%%%%%%%%%%%%%%%%%%%%%%%%%%%%%%%%%%%%%%%%%%%%%%%

%{\color{blue}
%En esta seccion trataremos los siguientes temas:
%\begin{enumerate}
% \item Desarrollaremos los resultados de Meyer y veremos como se aplican a los casos particulares de las matrices que consideramos.
% \item Haremos especial hincapie en el caso particular $m=2$, dando las formulas explicitas, ya que en este caso queda todo redondo.
% \item Se puede reinterpretar el caso general como una especie de {\sl caso de dos capas}, donde una capa es la realidad y la otra es una mezcla de las demas (idea de Luis).
% \item Que pasa si no consideramos que las matrices $L_i$ son estocasticas por filas. Esto permitiria considerar otros casos no cubiertos hasta ahora (como $\B_0$ o el dado por la representacion monoplex que tiene matriz de adyacencia $\tilde A$.
%\end{enumerate}
%De esta seccion se pueden encargar principalmente Esther y Julio.
%}

%{\color{red}
%Los contenidos de esta seccion se incluiran en el archivo {\bf eigenvector.tex} en lugar de incluirlos en el archivo Perron\_Multiplex-002.tex, para evitar que tengamos versiones diferentes.
%}

% Aqui van los contenidos de la seccion de relaciones del autovector grande con los de los bloques
% NO HACE FALTA PONER NINGUNA CABECERA NI NADA

Our approach is based on the Perron complementation method for finding the Perron eigenvector of a nonnegative irreducible matrix $A_{m\times m}$ with spectral radius $\rho$, see \cite{Meyer89}.  This method consists of uncoupling $A$ into smaller matrices whose Perron eigenvectors are coupled together in order to recover the Perron eigenvector of $A$ and it is described in Appendix \ref{annex}.  The Perron eigenvector $\pi=\left(\begin{array}{c}
\pi^1 \\ \hline
\pi^2 \\ \hline
\vdots \\ \hline
\pi^k
\end{array}
\right)>0
$ of each of  $\B_1,\B_2,\B_3$  is of the form 
 $\pi=\left(\begin{array}{c}
\xi^1p_1 \\ \hline
\xi^2p_2 \\ \hline
\vdots \\ \hline
\xi^kp_k
\end{array}
\right)>0
$
where each $p_i$ is the Perron eigenvector of the Perron complement $P_{ii}$, and will be calculated for all the three cases, and the normalizing scalars  or  coupling factors $\xi_i$ turn to be the $i^{\rm th}$-components of the Perron eigenvector of $W^t.$

Our only assumption is that $W$ is row-stochastic and that no $i^{\rm th}$-row of $W$ equals the $i^{\rm th}$-vector of the canonical basis $e_i$ of $\mathbb{R}^m$ (this means that all layers have influence at least on some other layer).

\noindent\underbar{Block matrix of type $\B_1$}:
The obtention of the Perron eigenvector $\pi$ of $\B_1$  follows from combining the $p_i's$ with the coupling factor, which is the Perron eigenvector of $W^t$. Remember that
\[
\B_1=
\left(
\begin{array}{c|c|c|c}
w_{11}L_1^t    & w_{21}L_2^t & \cdots & w_{m1}L_m^t \\ \hline
w_{12}L_1^t & w_{22}L_2^t    & \cdots & w_{m2}L_m^t \\ \hline
\vdots & \vdots & \ddots & \vdots \\ \hline
w_{1m}L_1^t & w_{2m}L_2 ^t & \cdots & w_{mm}L_m^t
\end{array}
\right)\in \mathbb{R}^{nm\times nm}.
\]
Let us calculate the Perron eigenvector $p_1$ of the Perron complement $P_{11}$.
First calculate $\left(
\begin{array}{c}
Q_2\\ \hline
Q_3\\ \hline
\vdots  \\ \hline
Q_m
\end{array}
\right)$, which is an eigenvector associated to 1 of the matrix
$${\mathcal A_1}^{p_1}=w_{11}L+\tilde{W}_{11}^{(1)}-w_{11}L\tilde{W}_{11}^{(1)} +\left(\begin{array}{c}
w_{12}L_1^t \\ \hline
w_{13}L_1^t\\ \hline
\vdots  \\ \hline
w_{1m}L_1^t
\end{array}
\right)(w_{21}L_2^t\dots w_{m1}L_m^t)
$$
where $L=\left(
\begin{array}{c|c|c}
L_1^t  & \cdots & 0 \\ \hline
0 &L_1^t &  \cdots  \\ \hline
\vdots &  \ddots & \vdots \\ \hline
0 &\cdots & L_1^t
\end{array}
\right)
$
and
$\tilde{W}_{11}^{(1)}=\left(
\begin{array}{c|c|c}
w_{22}L_2^t     & \cdots & w_{m2}L_m^t \\ \hline
w_{23}L_2^t &  \cdots & w_{m2}L_m^t \\ \hline
\vdots &  \ddots & \vdots \\ \hline
w_{2m}L_2^t  & \cdots & w_{mm}L_m^t
\end{array}
\right).$

Once the $Q_i's$ are obtained  use
$$
\left(
\begin{array}{c}
w_{12}L_1^t \\ \hline
w_{13}L_1^t\\ \hline
\vdots  \\ \hline
w_{1m}L_1^t
\end{array}
\right)p_1=\left(Id-\left(
\begin{array}{c|c|c}
w_{22}L_2^t     & \cdots & w_{m2}L_m^t \\ \hline
w_{23}L_2^t &  \cdots & w_{m2}L_m^t \\ \hline
\vdots &  \ddots & \vdots \\ \hline
w_{2m}L_2^t  & \cdots & w_{mm}L_m^t
\end{array}
\right)\right)\left(
\begin{array}{c}
Q_2\\ \hline
Q_3\\ \hline
\vdots  \\ \hline
Q_m
\end{array}
\right)
$$
to get $L_1^tp_1$ (remember that some of the $w_{1i}\ne 0$), and then
the equality
$$
p_1=w_{11}L_1^t p_1+(w_{21}L_2^t\dots w_{m1}L_m^t)\left(
\begin{array}{c}
Q_2\\ \hline
Q_3\\ \hline
\vdots  \\ \hline
Q_m
\end{array}
\right)$$
to recover $p_1$.%

The remaining $p_i's$ are analogously calculated.

\noindent\underbar{Block matrix of type $\B_2$}: The obtention of the Perron eigenvector $\pi$ of $\B_2$  follows from combining the $p_i's$ with the coupling factor, which is the Perron eigenvector of $W^t$. Remember that
\[
\B_2=
\left(
\begin{array}{c|c|c|c}
w_{11}L_1^t    & w_{21}L_1^t & \cdots & w_{m1}L_1^t \\ \hline
w_{12}L_2^t & w_{22}L_2^t    & \cdots & w_{m2}L_2^t \\ \hline
\vdots & \vdots & \ddots & \vdots \\ \hline
w_{1m}L_m^t & w_{2m}L_m ^t & \cdots & w_{mm}L_m^t
\end{array}
\right)\in \mathbb{R}^{nm\times nm}.
\]
Let us calculate the Perron eigenvector $p_1$ of the Perron complement $P_{11}$. 
First calculate
 $\left(
\begin{array}{c}
Q_2\\ \hline
Q_3\\ \hline
\vdots  \\ \hline
Q_m
\end{array}
\right),
$
which is an eigenvector associated to 1 of the matrix
$${\mathcal A_2}^{p_1}=w_{11}L+\tilde{W}_{11}^{(2)}-w_{11}\tilde{W}_{11}^{(2)}L +\left(
\begin{array}{c}
w_{12}L_2^t \\ \hline
w_{13}L_3^t\\ \hline
\vdots  \\ \hline
w_{1m}L_m^t
\end{array}
\right)
(w_{21}L_1^t,\dots, w_{m1}L_1^t)
$$
where $L=\left(
\begin{array}{c|c|c}
L_1^t  & \cdots & 0 \\ \hline
0 &L_1^t &  \cdots  \\ \hline
\vdots &  \ddots & \vdots \\ \hline
0 &\cdots & L_1^t
\end{array}
\right)
$
and
$\tilde{W}_{11}^{(2)}=\left(
\begin{array}{c|c|c}
w_{22}L_2^t     & \cdots & w_{m2}L_2^t \\ \hline
w_{23}L_3^t &  \cdots & w_{m2}L_3^t \\ \hline
\vdots &  \ddots & \vdots \\ \hline
w_{2m}L_m^t  & \cdots & w_{mm}L_m^t
\end{array}
\right).$
%$$
%\left(Id-
%\begin{array}{c|c|c}
%w_{22}L_2^t     & \cdots & w_{m2}L_2^t \\ \hline
%w_{23}L_3^t &  \cdots & w_{m2}L_3^t \\ \hline
%\vdots &  \ddots & \vdots \\ \hline
%w_{2m}L_m^t  & \cdots & w_{mm}L_m^t
%\end{array}
%\right)
%(Id-w_{11}L_1^t)-\left(
%\begin{array}{c}
%w_{12}L_2^t \\ \hline
%w_{13}L_3^t\\ \hline
%\vdots  \\ \hline
%w_{1m}L_m^t
%\end{array}
%\right)
%(w_{21}L_1,\dots, w_{m1}L_1)
%$$
Once the $Q_i's$ are obtained,
$$
p_1=(w_{21}L_1,\dots, w_{m1}L_1)\left(
\begin{array}{c}
Q_2\\ \hline
Q_3\\ \hline
\vdots  \\ \hline
Q_m
\end{array}
\right).
$$
The remaining $p_i's$ are analogously obtained.
%
%The obtention of the Perron vector $\pi$ of $\B_2$  follows from combining the $p_i's$ with the coupling factor, which is the Perron vector of $W^t$.

\noindent\underbar{Block matrix of type $\B_3$}: The obtention of the Perron eigenvector $\pi$ of $\B_3$  follows from combining the $p_i's$ with the coupling factor, which is the Perron eigenvector of $W^t$. Remember that
\[
\B_3=
\left(
\begin{array}{c|c|c|c}
w_{11}L_1^t    & w_{21}Id & \cdots & w_{m1}Id \\ \hline
w_{12}Id & w_{22}L_2^t    & \cdots & w_{m2}Id \\ \hline
\vdots & \vdots & \ddots & \vdots \\ \hline
w_{1m}Id & w_{2m}Id & \cdots & w_{mm}L_m^t
\end{array}
\right)\in \mathbb{R}^{nm\times nm}.
\]
The calculation of the Perron eigenvector $p_1$ of the Perron complement $P_{11}$ can be done as follows:
%Then, multiplying in both sides by
%$
%\left(
%\begin{array}{c}
%w_{12}Id \\ \hline
%w_{13}Id\\ \hline
%\vdots  \\ \hline
%w_{1m}Id
%\end{array}
%\right)
%$
%and
%Using the change of variables
%
%$$
%\left(
%\begin{array}{c}
%Q_2 \\ \hline
%Q_3\\ \hline
%\vdots  \\ \hline
%Q_m
%\end{array}
%\right)=\left(Id-\left(
%\begin{array}{c|c|c|c}
%w_{22}L_2^t & w_{32}Id & \cdots & w_{m2}Id \\ \hline
%w_{23}Id & \cdots& \cdots & w_{m2}L_m^t \\ \hline
%\vdots &  \vdots & \ddots	 &\vdots \\ \hline
%w_{2m}Id  & w_{3m}Id& \cdots & w_{mm}L_m^t
%\end{array}
%\right)\right)^{(-1)}
%\left(
%\begin{array}{c}
%w_{12}Id \\ \hline
%w_{13}Id\\ \hline
%\vdots  \\ \hline
%w_{1m}Id
%\end{array}
%\right)p_1
%$$
calculate $\left(
\begin{array}{c}
Q_2 \\ \hline
Q_3\\ \hline
\vdots  \\ \hline
Q_m
\end{array}
\right),$ which is an eigenvector associated to 1 of the matrix
$${\mathcal A_3}^{p_1}=w_{11}L+\tilde{W}_{11}^{(3)}-w_{11}L\tilde{W}_{11}^{(3)} +\left(\begin{array}{c}
w_{12}Id \\ \hline
w_{13}Id\\ \hline
\vdots  \\ \hline
w_{1m}Id
\end{array}
\right)(w_{21}Id\dots w_{m1}Id)
$$ where
$L=\left(
\begin{array}{c|c|c}
L_1^t  & \cdots & 0 \\ \hline
0 &L_1^t &  \cdots  \\ \hline
\vdots &  \ddots & \vdots \\ \hline
0 &\cdots & L_1^t
\end{array}
\right)
$
and
$\tilde{W}_{11}^{(3)}=\left(
\begin{array}{c|c|c|c}
w_{22}L_2^t & w_{32}Id   & \cdots & w_{m2}Id \\ \hline
w_{23}Id &w_{33}L_3^t &  \cdots & w_{m2}Id \\ \hline
\vdots & \vdots & \ddots & \vdots \\ \hline
w_{2m}Id & w_{3m}Id & \cdots & w_{mm}L_m^t
\end{array}
\right).$

Once the $Q_i's$ are obtained  use
$$
\left(
\begin{array}{c}
w_{12}Id \\ \hline
w_{13}Id\\ \hline
\vdots  \\ \hline
w_{1m}Id
\end{array}
\right)p_1=\left(Id-\left(
\begin{array}{c|c|c|c}
w_{22}L_2^t & w_{32}Id & \cdots & w_{m2}Id \\ \hline
w_{23}Id & \cdots& \cdots & w_{m2}L_m^t \\ \hline
\vdots &  \vdots & \ddots	 &\vdots \\ \hline
w_{2m}Id  & w_{3m}Id& \cdots & w_{mm}L_m^t
\end{array}
\right)\right)\left(
\begin{array}{c}
Q_2 \\ \hline
Q_3\\ \hline
\vdots  \\ \hline
Q_m
\end{array}
\right)
$$
to recover $p_1$ (remember that some $w_{1j}\ne 0$).

The remaining $p_i's$ are analogously calculated.

% the change of variables above to recover $p_1$.
%The remaining $p_i's$ are analogously calculated. The obtention of the Perron vector $\pi$ of $\B_2$  follows from combining the $p_i's$ with the coupling factor, which is the Perron vector of $W^t$.

%Notice that $c_{ij}=\|w_{ji}L_j(p_j)\|_1=w_{ji}\|L_jp_j\|_1$
%
%Now, since  $w_{ji}L_j$ are column stochastic we have $(1,\dots,1)w_{ji}L_j=(1,\dots,1)$; thus
%$$\xi_ip_i=\overset{k}{\underset{j=1}{\sum}}c_{ij}\xi_jp_j=\overset{k}{\underset{j=1}{\sum}}(1,\dots,1)w_{ij}w_{ji}L_jp_j\xi_jp_i=\overset{k}{\underset{j=1}{\sum}}w_{ij}\xi_j(1,\dots,1)p_jp_i=\overset{k}{\underset{j=1}{\sum}}w_{ij}\xi_jp_i$$
%
%This implies that $\xi_i=\overset{k}{\underset{j=1}{\sum}}w_{ij}\xi_j$ as $p_i$ is a non-zero vector. In other words, $(\xi_1,\dots,\xi_k)$ is the Perron eigenvector of $W=(w_{ij})$

\bigskip
\bigskip

\bigskip
\subsection{Particular case of two layers ($m=2$)} 

We will show that the eigenvectors associated to the principal eigenvalue 1 can be computed in terms of the eigenvectors associated to 1 of certain matrices related to $L_1^t$, $L_2^t$ and the elements of $W$. The only assumption on $W$ is that it is row-stochastic. The details of the calculations will be shown in \S\ref{annex}.

\noindent\underbar{Block matrix of type $\B_1$, $m=2$}:

\[
\B_1=
\left(
\begin{array}{c|c}
w_{11}L_1^t    & w_{21}L_2^t  \\ \hline
w_{12}L_1^t & w_{22}L_2^t    \\ %\vdots & \vdots & \ddots & \vdots \\ \hline
%w_{1m}L_1^t & w_{2m}L_2 ^t & \cdots & w_{mm}L_m^t
\end{array}
\right),
\]
where $L_\ell^t$ is the transpose of the row normalization of the adjacency matrix of layer $S_\ell$.

\noindent (a) If both $w_{11}\ne 1$ and $w_{22}\ne 1$ then if $\left(
\begin{array}{c}
\pi_1 \\ \hline
\pi_2
\end{array}
\right)$  is an eigenvector associated to the eigenvalue 1,
%we have
%$$\left\{
%    \begin{array}{ll}
%      \pi_1=w_{11}L_1^t \pi_1+w_{21}L_2^t \pi_2,  \\
%      \pi_2=w_{12}L_1^t \pi_1+ w_{22}L_2^t \pi_2.
%    \end{array}
%  \right.
%$$ From here,
%taking into account that both $(I-w_{11}L_1^t)$ and $(I-w_{22}L_2^t)$ are invertible matrices, 
we get that
%$\pi_1=w_{21}(I-w_{11}L_1^t)^{-1}L_2^t \pi_2$, and  $\pi_2=w_{12}(I-w_{22}L_2^t)^1 L_1^t \pi_1.$ Substituting in the above equations we get
%$$\left\{
%    \begin{array}{ll}
%      \pi_1=(w_{11}L_1^t+w_{12}w_{21}L_2^t(I-w_{22}L_2^t)^{-1}L_1^t) \pi_1,  \\
%      \pi_2=(w_{22}L_2^t+w_{12}w_{21}L_1^t(I-w_{11}L_1^t)^{-1}L_2^t)\pi_2.
%    \end{array}
%  \right.
%$$  Now multiplying the first equation by the matrix $(I-w_{22}L_2^t)$ on the left, and the second equation by the matrix $(I-w_{11}L_1^t)$ on the left we get
%$$\left\{
%    \begin{array}{ll}
%      \pi_1=(w_{11}L_1^t+ w_{22}L_2^t+(1-w_{11}-w_{22})L_2^tL_1^t)\pi_1,  \\
%      \pi_2=(w_{11}L_1^t+ w_{22}L_2^t+(1-w_{11}-w_{22})L_1^tL_2^t)\pi_2,
%    \end{array}
%  \right.
%$$
%i.e.,
$\pi_1$ and $\pi_2$ are eigenvectors associated to 1 to the column stochastic matrices
$$\begin{array}{ll}
      {\mathcal A}_1^{\pi_1}=(w_{11}L_1^t+ w_{22}L_2^t+(1-w_{11}-w_{22})L_2^tL_1^t), \hbox{ and} \\
      {\mathcal A}_1^{\pi_2}=(w_{11}L_1^t+ w_{22}L_2^t+(1-w_{11}-w_{22})L_1^tL_2^t).
    \end{array}
$$

\noindent (b) If $w_{11}=1$ then $w_{12}=0$
%in which case $\B_1$ is of the form,
%\[
%\B_1=
%\left(
%\begin{array}{c|c}
%L_1^t    & w_{21}L_2^t  \\ \hline
%0 & w_{22}L_2^t    \\ %\vdots & \vdots & \ddots & \vdots \\ \hline
%%w_{1m}L_1^t & w_{2m}L_2 ^t & \cdots & w_{mm}L_m^t
%\end{array}
%\right),
%\]
and if the vector
$\left(
\begin{array}{c}
\pi_1 \\ \hline
\pi_2
\end{array}
\right)$ is associated to the eigenvalue 1 then
%$$\left\{
%    \begin{array}{ll}
%      \pi_1=L_1^t \pi_1+w_{21}L_2^t \pi_2,  \\
%      \pi_2=w_{22}L_2^t \pi_2.
%    \end{array}
%  \right.
%$$
we have one of the three following situations:

\noindent (b.1) $0<w_{22}<1$: %in this case $\pi_2=0$ since $L_2^t$ is column stochastic and cannot have nonzero eigenvectors with associated to an eigenvalue $1/w_{22}>1$. Therefore
the eigenvectors associated to 1 of $\B_1$ have the form
$\left(
\begin{array}{c}
\pi_1 \\ \hline
0
\end{array}
\right)$
where $\pi_1$ is an eigenvector of $L_1^t$ associated to 1.

\noindent (b.2) $w_{22}=0$: %in this case $w_{21}=1$ and we have that
the eigenvectors associated to $\B_1$ have the form
$\left(
\begin{array}{c}
\pi_1 \\ \hline
0
\end{array}
\right)$
where $\pi_1$ is an eigenvector of $L_1^t$ associated to 1.

\noindent (b.3) $w_{22}=1$:
%in this case $W$ is the identity (there is no influence of a layer into another layer) and
the eigenvectors of $\B_1$ associated to 1 have the form  $\left(
\begin{array}{c}
\pi_1 \\ \hline
\pi_2
\end{array}
\right)$ where $\pi_1$ is an eigenvector of $L_1^t$ associated to 1 and $\pi_2$ an eigenvector of $L_2^t$ associated to 1.

\noindent (c) If $w_{22}=1$ then, arguing as in case (b) either $w_{11}=1$ and we are again in the situation of (b.3) or the  eigenvector of $\B_1$ associated to 1 are of the form $\left(
\begin{array}{c}
0    \\ \hline
\pi_2
\end{array}
\right)$ where $\pi_2$ is an eigenvector of $L_2^t$ associated to 1.

\noindent\underbar{Block matrix of type $\B_2$, $m=2$}:

\[
\B_2=
\left(
\begin{array}{c|c}
w_{11}L_1^t    & w_{21}L_1^t  \\ \hline
w_{12}L_2^t & w_{22}L_2^t    \\ %\vdots & \vdots & \ddots & \vdots \\ \hline
%w_{1m}L_1^t & w_{2m}L_2 ^t & \cdots & w_{mm}L_m^t
\end{array}
\right),
\]
where $L_\ell^t$ is the transpose of the row normalization of the adjacency matrix of layer $S_\ell$.

\noindent (a) If both $w_{11}\ne 1$ and $w_{22}\ne 1$ then if $\left(
\begin{array}{c}
\pi_1 \\ \hline
\pi_2
\end{array}
\right)$  is an eigenvector associated to the eigenvalue 1
%we have
%$$\left\{
%    \begin{array}{ll}
%      \pi_1=w_{11}L_1^t \pi_1+w_{21}L_1^t \pi_2,  \\
%      \pi_2=w_{12}L_2^t \pi_1+ w_{22}L_2^t \pi_2.
%    \end{array}
%  \right.
%$$ From here,
%taking into account that both $(I-w_{11}L_1^t)$ and $(I-w_{22}L_2^t)$ are invertible matrices,
%we get that $\pi_1=w_{21}(I-w_{11}L_1^t)^{-1}L_1^t \pi_2$, and  $\pi_2=w_{12}(I-w_{22}L_2^t)^{-1} L_2^t \pi_1.$ Substituting in the above equations we get
%$$\left\{
%    \begin{array}{ll}
%     (I-w_{11}L_1^t) \pi_1=w_{12}w_{21}L_1^t(I-w_{22}L_2^t)^{-1}L_2^t \pi_1,  \\
%     (I-w_{22}L_2^t) \pi_2=w_{12}w_{21}L_2^t(I-w_{11}L_1^t)^{-1}L_1^t \pi_2,
%    \end{array}
%  \right.
%$$
%so using that $L_1^t$ and $(I-w_{11}L_1^t)^{-1}$ commute and $L_2^t$ and $(I-w_{22}L_2^t)^{-1}$ commute we have
%$$\left\{
%    \begin{array}{ll}
%     \pi_1=w_{12}w_{21}L_1^t(I-w_{11}L_1^t)^{-1} (I-w_{22}L_2^t)^{-1}L_2^t \pi_1,  \\
%     \pi_2=w_{12}w_{21}L_2^t(I-w_{22}L_2^t)^{-1}(I-w_{11}L_1^t)^{-1}L_1^t\pi_2.
%    \end{array}
%  \right. \eqno{(1)}
%$$
and defining $\pi_1^{aux}=(I-w_{11}L_1^t)^{-1} (I-w_{22}L_2^t)^{-1}L_2^t \pi_1$ and $\pi_2^{aux}=(I-w_{22}L_2^t)^{-1}(I-w_{11}L_1^t)^{-1}L_1^t\pi_2$, we get that
%. By the   equations (1)
%$$\left\{
%    \begin{array}{ll}
%     \pi_1=w_{12}w_{21}L_1^t \pi_1^{aux},  \\
%     \pi_2=w_{12}w_{21}L_2^t \pi_2^{aux},
%    \end{array}
%  \right.\eqno{(2)}
%$$
%and from (1) and (2)
%$$\left\{
%    \begin{array}{ll}
%     (I-w_{22}L_2^t)(I-w_{11}L_1^t) \pi_1^{aux}=L_2^t \pi_1= L_2^tw_{12}w_{21}L_1^t \pi_1^{aux},  \\
%     (I-w_{22}L_2^t)(I-w_{11}L_1^t) \pi_2^{aux}=L_1^t \pi_2= L_1^t w_{12}w_{21}L_2^t \pi_2^{aux},
%     \end{array}
%  \right.
%$$
%i.e.,
$\pi_1^{aux}$ and $\pi_2^{aux}$ are eigenvectors associated to 1 of the column stochastic matrices
$$\begin{array}{ll}
      {\mathcal A}_2^{\pi_1^{aux}}=(w_{11}L_1^t+w_{22}L_2^t-w_{11}w_{22}L_2^tL_1^t+w_{12}w_{21}L_2^tL_1^t), \hbox{ and} \\
      {\mathcal A}_2^{\pi_2^{aux}}=(w_{11}L_1^t+w_{22}L_2^t-w_{11}w_{22}L_1^tL_2^t+w_{12}w_{21}L_1^tL_2^t).
    \end{array}
$$ After computing $\pi_1^{aux}$ and $\pi_2^{aux}$,
$$\left\{
    \begin{array}{ll}
     \pi_1=w_{12}w_{21}L_1^t \pi_1^{aux},  \\
     \pi_2=w_{12}w_{21}L_2^t \pi_2^{aux}.
    \end{array}
  \right.
$$

\noindent (b) ($w_{11}=1$) and (c)  ($w_{22}=1$) give the same results as for matrices of type $\B_1$.

\noindent\underbar{Block matrix of type $\B_3$, $m=2$}:

\[
\B_2=
\left(
\begin{array}{c|c}
w_{11}L_1^t    & w_{21}I_2  \\ \hline
w_{12}I_2 & w_{22}L_2^t    \\ %\vdots & \vdots & \ddots & \vdots \\ \hline
%w_{1m}L_1^t & w_{2m}L_2 ^t & \cdots & w_{mm}L_m^t
\end{array}
\right),
\]
where $L_\ell^t$ is the transpose of the row normalization of the adjacency matrix of layer $S_\ell$.

\noindent (a) If both $w_{11}\ne 1$ and $w_{22}\ne 1$ then if $\left(
\begin{array}{c}
\pi_1 \\ \hline
\pi_2
\end{array}
\right)$  is an eigenvector associated to the eigenvalue 1,
%we have
%$$\left\{
%    \begin{array}{ll}
%      \pi_1=w_{11}L_1^t \pi_1+w_{21} \pi_2,  \\
%      \pi_2=w_{12} \pi_1+ w_{22}L_2^t \pi_2.
%    \end{array}
%  \right.
%$$
%taking into account that both $(I-w_{11}L_1^t)$ and $(I-w_{22}L_2^t)$ are invertible matrices,
we get that
%$\pi_1=w_{21}(I-w_{11}L_1^t)^{-1} \pi_2$, and  $\pi_2=w_{12}(I-w_{22}L_2^t)^{-1}  \pi_1.$ Substituting in the above equations we get
%$$\left\{
%    \begin{array}{ll}
%     (I-w_{11}L_1^t) \pi_1=w_{12}w_{21}(I-w_{22}L_2^t)^{-1} \pi_1,  \\
%     (I-w_{22}L_2^t) \pi_2=w_{12}w_{21}(I-w_{11}L_1^t)^{-1} \pi_2,
%    \end{array}
%  \right.
%$$
%so multiplying in both sides by $(I-w_{11}L_1^t)$ and $(I-w_{22}L_2^t)$ respectively  we have
%$$\left\{
%    \begin{array}{ll}
%     w_{12}w_{21}\pi_1=(I-w_{22}L_2^t)(I-w_{11}L_1^t)\pi_1,  \\
%     w_{12}w_{21}\pi_2=(I-w_{11}L_1^t)(I-w_{22}L_2^t)\pi_2.
%    \end{array}
%  \right.
%$$
%Therefore,
 $\pi_1$ and $\pi_2$ are eigenvectors associated to 1 to the column stochastic matrices
$$
\begin{array}{ll}
      {\mathcal A}_2^{\pi_1}=(w_{11}L_1^t+w_{22}L_2^t-w_{11}w_{22}L_2^tL_1^t+ w_{12}w_{21}I_2), \hbox{ and} \\
      {\mathcal A}_2^{\pi_2}=(w_{11}L_1^t+w_{22}L_2^t-w_{11}w_{22}L_1^tL_2^t+w_{12}w_{21}I_2).
    \end{array}
$$

\noindent (b) ($w_{11}=1$) and (c)  ($w_{22}=1$) give the same results as for matrices of type $\B_1$.

%In what follows we extend the previous work to higher dimensions. As expected the formulas obtained are not as simple.

%\date{\today}

%\end{document}

% Final del documento de relaciones del autovector grande con los de los bloques 

%%%%%%%%%%%%%%%%%%%%%%%%%%%%%%%%%%%%%%%%%%%%%%%%%%%%%%%%%%%%%%%%%%%%%

%%%%%%%%%%%%%%%%%%%%%%%%%%%%%%%%%%%%%%%%%%%%%%%%%%%%%%%%%%%%%%%%%%%%%

\newpage

%%%%%%%%%%%%%%%%%%%%%%%%%%%%%%%%%%%%%%%%%%%%%%%%%%%%%%%%%%%%%%%%%%%%%
\appendix
 \section{Mathematical proof of the results of section 4}
 \label{annex}
%%%%%%%%%%%%%%%%%%%%%%%%%%%%%%%%%%%%%%%%%%%%%%%%%%%%%%%%%%%%%%%%%%%%%

%
%{\color{blue}
%En esta seccion incluiremos las demostraciones largas y las cuentas para que los {\sl no matematicos} no se asusten demasiado.
%}
%
%
%{\color{red}
%Los contenidos de esta seccion se incluiran en el archivo {\bf anexo.tex} en lugar de incluirlos en el archivo Perron\_Multiplex-002.tex, para evitar que tengamos versiones diferentes.
%}

% Aqui van las demostraciones largas y contenidos para el anexo
% NO HACE FALTA PONER NINGUNA CABECERA NI NADA

\noindent\underbar{Perron complementation method} for finding the Perron vector of a nonnegative irreducible matrix $A_{m\times m}$ with spectral radius $\rho$ (\cite{Meyer89}): This method consists of uncoupling $A$ into smaller matrices whose Perron vectors are coupled together in order to recover the Perron vector of $A$. Let us briefly recall it:

Given a $k$-level partition

 $$A=\left(
\begin{array}{c|c|c|c}
A_{11}    & A_{12} & \cdots & A_{1k} \\ \hline
A_{21} & A_{22}    & \cdots & A_{2k} \\ \hline
\vdots & \vdots & \ddots & \vdots \\ \hline
A_{k1} & A_{k2} & \cdots & A_{kk}
\end{array}
\right)$$
where all the diagonal blocks $A_{ii}$ are square, we consider  the principal block submatrices $A_i$ of $A$ obtained by deleting the $i^{\rm th}$-row of blocks and the $i^{\rm th}$-column of blocks from $A$. We also consider $$A_{i*}=(A_{i1} A_{i2}\cdots A_{i,i-1}\  A_{i,i+1}\cdots A_{ik})$$ and

$$A_{*i}=\left(\begin{array}{c}
A_{1i} \\ \hline
\vdots \\ \hline
A_{i-1,i} \\ \hline
A_{i+1,i} \\ \hline
\vdots \\ \hline
A_{ki}
\end{array}
\right).
$$

The Perron complement of $A_{ii}$ in $A$ is defined as the matrix
$$
P_{ii}=A_{ii}+A_{i*}(\rho Id-A_i)^{-1}A_{*i}.
$$

The importance of the Perron complements stems from the fact that if $A$ is nonnegative and irreducible with spectral radius $\rho$, then $P_{ii}$ is also nonnegative and irreducible with spectral radius $\rho$. In addition, if
$\pi=\left(\begin{array}{c}
\pi^1 \\ \hline
\pi^2 \\ \hline
\vdots \\ \hline
\pi^k
\end{array}
\right)>0
$
is the Perron vector of $A$, partitioned accordingly, then $P_{ii}\pi^i=\rho\pi^i$,
that is, $\pi^i$ is a positive eigenvector of $P_{ii}$ associated to $\rho$ (\cite[Thm 2.1 and 2.2]{Meyer89}).
Call  $p_i\equiv\dfrac{\pi^i}{\|\pi^i\|_1}$, the Perron vector of $P_{ii}$. The normalizing scalar $\xi^i\equiv \|\pi^i\|_1$, or \emph{coupling factor}, turns out to be the $i^{\rm th}$-component of the Perron eigenvector
$\left(\begin{array}{c}
\xi^1 \\ \hline
\xi^2 \\ \hline
\vdots \\ \hline
\xi^k
\end{array}
\right)$
of the \emph{coupling matrix} $C\equiv (c_{ij})$ , where $c_{ij}=\|A_{ij}p_j\|_1$. Thus, the Perron vector $\pi$ can be expressed as
$\pi=\left(\begin{array}{c}
\xi^1p_1 \\ \hline
\xi^2 p_2\\ \hline
\vdots \\ \hline
\xi^k p_k
\end{array}
\right).$

Our immediate task is to identify the Perron complements for each of the three types of matrices considered and proceed accordingly. Each $L_\ell$ is row stochastic and therefore $L_\ell^t$ is column stochastic; similarly $W$ is row stochastic,  hence  each of the matrices $\B_1, \B_2$ and $\B_3$ given in Section 1 is also column stochastic and its maximal eigenvalue is one.

It will be assumed that  no $i^{\rm th}$-row of $W$ equals the $i^{\rm th}$-vector of the canonical basis $e_i$ of $\mathbb{R}^m$ (this means that all layers have influence at least on some other layer). 

As for the coupling matrix $C$, since $L_\ell^t$ are column stochastic, in each of the three cases we get that $C=W^t$ and therefore the coupling factors correspond to the Perron eigenvector of $W^t$.
%:
%notice that $c_{ij}=\|w_{ji}L_jp_j\|_1=w_{ji}\|L_jp_j\|_1$.
%Now, since  $w_{ji}L_j$ are column stochastic we have $(1,\dots,1)w_{ji}L_j=(1,\dots,1)$; thus
%\begin{align*}\xi_ip_i&=\overset{k}{\underset{j=1}{\sum}}c_{ij}\xi_jp_j=
%\overset{k}{\underset{j=1}{\sum}}(1,\dots,1)w_{ij}w_{ji}L_jp_j\xi_jp_i\\&=\overset{k}{\underset{j=1}{\sum}}w_{ij}\xi_j(1,\dots,1)p_jp_i=\overset{k}{\underset{j=1}{\sum}}w_{ij}\xi_jp_i.
%\end{align*}
%This implies that $\xi_i=\overset{k}{\underset{j=1}{\sum}}w_{ij}\xi_j$ as $p_i$ is a non-zero vector. In other words, $(\xi_1,\dots,\xi_k)$ is the Perron eigenvector of $W=(w_{ij})$.

\noindent\underbar{Block matrix of type $\B_1$}: The obtention of the Perron vector $\pi$ of $\B_1$  follows from combining the $p_i's$ with the coupling factor, which is the Perron vector of $W^t$.

\[
\B_1=
\left(
\begin{array}{c|c|c|c}
w_{11}L_1^t    & w_{21}L_2^t & \cdots & w_{m1}L_m^t \\ \hline
w_{12}L_1^t & w_{22}L_2^t    & \cdots & w_{m2}L_m^t \\ \hline
\vdots & \vdots & \ddots & \vdots \\ \hline
w_{1m}L_1^t & w_{2m}L_2 ^t & \cdots & w_{mm}L_m^t
\end{array}
\right)\in \mathbb{R}^{nm\times nm}.
\]

Let us calculate the Perron vector $p_1$ of the Perron complement $P_{11}$. It satisfies
\begin{align*}
&p_1=w_{11}L_1^t p_1+(w_{21}L_2^t\dots w_{m1}L_m^t)
\left(Id-\left(
\begin{array}{c|c|c}
w_{22}L_2^t     & \cdots & w_{m2}L_m^t \\ \hline
w_{23}L_2^t &  \cdots & w_{m2}L_m^t \\ \hline
\vdots &  \ddots & \vdots \\ \hline
w_{2m}L_2^t  & \cdots & w_{mm}L_m^t
\end{array}
\right)\right)^{-1}
\left(
\begin{array}{c}
w_{12}L_1^t \\ \hline
w_{13}L_1^t\\ \hline
\vdots  \\ \hline
w_{1m}L_1^t
\end{array}
\right)p_1
\end{align*}

Then
$$
\left(
\begin{array}{c}
w_{12}L_1^t \\ \hline
w_{13}L_1^t\\ \hline
\vdots  \\ \hline
w_{1m}L_1^t
\end{array}
\right)p_1
=w_{11}\left(\begin{array}{c}
w_{12}L_1^t \\ \hline
w_{13}L_1^t\\ \hline
\vdots  \\ \hline
w_{1m}L_1^t
\end{array}
\right)
L_1^tp_1+
\left(
\begin{array}{c}
w_{12}L_1^t \\ \hline
w_{13}L_1^t\\ \hline
\vdots  \\ \hline
w_{1m}L_1^t
\end{array}
\right)(w_{21}L_2^t\dots w_{m1}L_m^t)\left(
\begin{array}{c}
Q_2\\ \hline
Q_3\\ \hline
\vdots  \\ \hline
Q_m
\end{array}
\right)
$$
where the following change of variables is used
$$\left(
\begin{array}{c}
Q_2\\ \hline
Q_3\\ \hline
\vdots  \\ \hline
Q_m
\end{array}
\right)=
\left(Id-\left(
\begin{array}{c|c|c}
w_{22}L_2^t     & \cdots & w_{m2}L_m^t \\ \hline
w_{23}L_2^t &  \cdots & w_{m2}L_m^t \\ \hline
\vdots &  \ddots & \vdots \\ \hline
w_{2m}L_2^t  & \cdots & w_{mm}L_m^t
\end{array}
\right)\right)^{-1}
\left(
\begin{array}{c}
w_{12}L_1^t \\ \hline
w_{13}L_1^t\\ \hline
\vdots  \\ \hline
w_{1m}L_1^t
\end{array}
\right)p_1.
$$
Equivalently

$
\left(Id-\left(
\begin{array}{c|c|c}
w_{22}L_2^t     & \cdots & w_{m2}L_m^t \\ \hline
w_{23}L_2^t &  \cdots & w_{m2}L_m^t \\ \hline
\vdots &  \ddots & \vdots \\ \hline
w_{2m}L_2^t  & \cdots & w_{mm}L_m^t
\end{array}
\right)\right)\left(
\begin{array}{c}
Q_2\\ \hline
Q_3\\ \hline
\vdots  \\ \hline
Q_m
\end{array}
\right)=
w_{11}
\left(
\begin{array}{c|c|c|c}
L_1^t &0    & \cdots & 0 \\ \hline
0 &L_1^t &  \cdots & 0 \\ \hline
\vdots & \vdots & \ddots & \vdots \\ \hline
0 & 0&\cdots & L_1^t
\end{array}
\right)
\left(\begin{array}{c}
w_{12}L_1^t \\ \hline
w_{13}L_1^t\\ \hline
\vdots  \\ \hline
w_{1m}L_1^t
\end{array}
\right)
p_1+
\left(
\begin{array}{c}
w_{12}L_1^t \\ \hline
w_{13}L_1^t\\ \hline
\vdots  \\ \hline
w_{1m}L_1^t
\end{array}
\right)(w_{21}L_2^t\dots w_{m1}L_m^t)
\left(
\begin{array}{c}
Q_2\\ \hline
Q_3\\ \hline
\vdots  \\ \hline
Q_m
\end{array}
\right),
$
\vspace{5mm}

or
\vspace{5mm}

$
\left(Id-
\begin{array}{c|c|c}
w_{22}L_2^t     & \cdots & w_{m2}L_m^t \\ \hline
w_{23}L_2^t &  \cdots & w_{m2}L_m^t \\ \hline
\vdots &  \ddots & \vdots \\ \hline
w_{2m}L_2^t  & \cdots & w_{mm}L_m^t
\end{array}
\right)\left(
\begin{array}{c}
Q_2\\ \hline
Q_3\\ \hline
\vdots  \\ \hline
Q_m
\end{array}
\right)=\\
w_{11}
\left(
\begin{array}{c|c|c|c}
L_1^t &0    & \cdots & 0 \\ \hline
0 &L_1^t &  \cdots & 0 \\ \hline
\vdots & \vdots & \ddots & \vdots \\ \hline
0 & 0&\cdots & L_1^t
\end{array}
\right)
\left(Id-
\begin{array}{c|c|c}
w_{22}L_2^t     & \cdots & w_{m2}L_m^t \\ \hline
w_{23}L_2^t &  \cdots & w_{m2}L_m^t \\ \hline
\vdots &  \ddots & \vdots \\ \hline
w_{2m}L_2^t  & \cdots & w_{mm}L_m^t
\end{array}
\right)\left(
\begin{array}{c}
Q_2\\ \hline
Q_3\\ \hline
\vdots  \\ \hline
Q_m
\end{array}
\right)
+\\
+\left(
\begin{array}{c}
w_{12}L_1^t \\ \hline
w_{13}L_1^t\\ \hline
\vdots  \\ \hline
w_{1m}L_1^t
\end{array}
\right)(w_{21}L_2^t\dots w_{m1}L_m^t)
\left(
\begin{array}{c}
Q_2\\ \hline
Q_3\\ \hline
\vdots  \\ \hline
Q_m
\end{array}\right).
$
\vspace{2mm}
This is equivalent to
$\left(
\begin{array}{c}
Q_2\\ \hline
Q_3\\ \hline
\vdots  \\ \hline
Q_m
\end{array}
\right)$ being an eigenvector associated to 1 of the matrix

$${\mathcal A_1}^{p_1}=w_{11}L+\tilde{W}_{11}^{(1)}-w_{11}L\tilde{W}_{11}^{(1)} +\left(\begin{array}{c}
w_{12}L_1^t \\ \hline
w_{13}L_1^t\\ \hline
\vdots  \\ \hline
w_{1m}L_1^t
\end{array}
\right)(w_{21}L_2^t\dots w_{m1}L_m^t)
$$

where $L=\left(
\begin{array}{c|c|c}
L_1^t  & \cdots & 0 \\ \hline
0 &L_1^t &  \cdots  \\ \hline
\vdots &  \ddots & \vdots \\ \hline
0 &\cdots & L_1^t
\end{array}
\right)
$
and
$\tilde{W}_{11}^{(1)}=\left(
\begin{array}{c|c|c}
w_{22}L_2^t     & \cdots & w_{m2}L_m^t \\ \hline
w_{23}L_2^t &  \cdots & w_{m2}L_m^t \\ \hline
\vdots &  \ddots & \vdots \\ \hline
w_{2m}L_2^t  & \cdots & w_{mm}L_m^t
\end{array}
\right).$
Once the $Q_i's$ are obtained we use
$$
\left(
\begin{array}{c}
w_{12}L_1^t \\ \hline
w_{13}L_1^t\\ \hline
\vdots  \\ \hline
w_{1m}L_1^t
\end{array}
\right)p_1=\left(Id-\left(
\begin{array}{c|c|c}
w_{22}L_2^t     & \cdots & w_{m2}L_m^t \\ \hline
w_{23}L_2^t &  \cdots & w_{m2}L_m^t \\ \hline
\vdots &  \ddots & \vdots \\ \hline
w_{2m}L_2^t  & \cdots & w_{mm}L_m^t
\end{array}
\right)\right)\left(
\begin{array}{c}
Q_2\\ \hline
Q_3\\ \hline
\vdots  \\ \hline
Q_m
\end{array}
\right)
$$
to get $L_1^tp_1$ (since some $w_{1i}\ne 0$), and then
the equality
$$
p_1=w_{11}L_1^t p_1+(w_{21}L_2^t\dots w_{m1}L_m^t)\left(
\begin{array}{c}
Q_2\\ \hline
Q_3\\ \hline
\vdots  \\ \hline
Q_m
\end{array}
\right)$$
to recover $p_1$.

The remaining $p_i's$ are analogously calculated.

\noindent\underbar{Block matrix of type $\B_2$}:
The obtention of the Perron vector $\pi$ of $\B_2$  follows from combining the $p_i's$ with the coupling factor, which is the Perron vector of $W^t$.
\[
\B_2=
\left(
\begin{array}{c|c|c|c}
w_{11}L_1^t    & w_{21}L_1^t & \cdots & w_{m1}L_1^t \\ \hline
w_{12}L_2^t & w_{22}L_2^t    & \cdots & w_{m2}L_2^t \\ \hline
\vdots & \vdots & \ddots & \vdots \\ \hline
w_{1m}L_m^t & w_{2m}L_m ^t & \cdots & w_{mm}L_m^t
\end{array}
\right)\in \mathbb{R}^{nm\times nm}.
\]

Let us calculate the Perron vector $p_1$ of the Perron complement $P_{11}$. It satisfies
$$
p_1=w_{11}L_1^t p_1+(w_{21}L_1^t\dots w_{m1}L_1^t)
\left(Id-\left(
\begin{array}{c|c|c}
w_{22}L_2^t     & \cdots & w_{m2}L_2^t \\ \hline
w_{23}L_3^t &  \cdots & w_{m2}L_3^t \\ \hline
\vdots &  \ddots & \vdots \\ \hline
w_{2m}L_m^t  & \cdots & w_{mm}L_m^t
\end{array}
\right)\right)^{-1}
\left(
\begin{array}{c}
w_{12}L_2^t \\ \hline
w_{13}L_3^t\\ \hline
\vdots  \\ \hline
w_{1m}L_m^t
\end{array}
\right)p_1
$$
so 
$(Id-w_{11}L_1^t)p_1=(w_{21}L_1^t\dots w_{m1}L_1^t)
\left(Id-\left(
\begin{array}{c|c|c}
w_{22}L_2^t     & \cdots & w_{m2}L_2^t \\ \hline
w_{23}L_3^t &  \cdots & w_{m2}L_3^t \\ \hline
\vdots &  \ddots & \vdots \\ \hline
w_{2m}L_m^t  & \cdots & w_{mm}L_m^t
\end{array}
\right)\right)^{-1}
\left(
\begin{array}{c}
w_{12}L_2^t \\ \hline
w_{13}L_3^t\\ \hline
\vdots  \\ \hline
w_{1m}L_m^t
\end{array}
\right)p_1
$
or, as $w_{11}\neq 1$,
$$p_1=(Id-w_{11}L_1^t)^{-1}(w_{21}L_1^t\dots w_{m1}L_1^t)
\left(Id-\left(
\begin{array}{c|c|c}
w_{22}L_2^t     & \cdots & w_{m2}L_2^t \\ \hline
w_{23}L_3^t &  \cdots & w_{m2}L_3^t \\ \hline
\vdots &  \ddots & \vdots \\ \hline
w_{2m}L_m^t  & \cdots & w_{mm}L_m^t
\end{array}
\right)\right)^{-1}
\left(
\begin{array}{c}
w_{12}L_2^t \\ \hline
w_{13}L_3^t\\ \hline
\vdots  \\ \hline
w_{1m}L_m^t
\end{array}
\right)p_1.
$$
Now, calling
$$
\tilde{C}=Id-
\left(\begin{array}{c|c|c|c}
(Id-w_{11}L_1^t)^{-1} &0    & \cdots & 0 \\ \hline
0 &  (Id-w_{11}L_1^t)^{-1} &\cdots & 0 \\ \hline
\vdots &  \vdots & \ddots &\vdots\\ \hline
0 &0 & \cdots & (Id-w_{11}L_1^t)^{-1}
\end{array}
\right)
$$
we get by matrix commutation
$$p_1=(w_{21}L_1^t\dots w_{m1}L_1^t)\
\tilde{C}\left(Id-\left(
\begin{array}{c|c|c}
w_{22}L_2^t     & \cdots & w_{m2}L_2^t \\ \hline
w_{23}L_3^t &  \cdots & w_{m2}L_3^t \\ \hline
\vdots &  \ddots & \vdots \\ \hline
w_{2m}L_m^t  & \cdots & w_{mm}L_m^t
\end{array}
\right)\right)^{-1}
\left(
\begin{array}{c}
w_{12}L_2^t \\ \hline
w_{13}L_3^t\\ \hline
\vdots  \\ \hline
w_{1m}L_m^t
\end{array}
\right)p_1.
$$
Multiplyig in both sides by
$
\left(
\begin{array}{c}
w_{12}L_2^t\\ \hline
w_{13}L_3^t\\ \hline
\vdots  \\ \hline
w_{1m}L_m^t
\end{array}
\right)
$
and using the change of variables
$\left(
\begin{array}{c}
Q_2\\ \hline
Q_3\\ \hline
\vdots  \\ \hline
Q_m
\end{array}
\right)\equiv\tilde{C}\left(Id-
\begin{array}{c|c|c}
w_{22}L_2^t     & \cdots & w_{m2}L_2^t \\ \hline
w_{23}L_3^t &  \cdots & w_{m2}L_3^t \\ \hline
\vdots &  \ddots & \vdots \\ \hline
w_{2m}L_m^t  & \cdots & w_{mm}L_m^t
\end{array}
\right)^{-1}
\left(
\begin{array}{c}
w_{12}L_2^t \\ \hline
w_{13}L_3^t\\ \hline
\vdots  \\ \hline
w_{1m}L_m^t
\end{array}
\right)p_1$
\vspace{3mm}
we get that
% in other words
$\left(
\begin{array}{c}
Q_2\\ \hline
Q_3\\ \hline
\vdots  \\ \hline
Q_m
\end{array}
\right)
$
is an eigenvector associated to 1 of the matrix
$${\mathcal A_2}^{p_1}=w_{11}L+\tilde{W}_{11}^{(2)}-w_{11}\tilde{W}_{11}^{(2)}L +\left(
\begin{array}{c}
w_{12}L_2^t \\ \hline
w_{13}L_3^t\\ \hline
\vdots  \\ \hline
w_{1m}L_m^t
\end{array}
\right)
(w_{21}L_1^t,\dots, w_{m1}L_1^t)
$$
where $L=\left(
\begin{array}{c|c|c}
L_1^t  & \cdots & 0 \\ \hline
0 &L_1^t &  \cdots  \\ \hline
\vdots &  \ddots & \vdots \\ \hline
0 &\cdots & L_1^t
\end{array}
\right)
$
and
$\tilde{W}_{11}^{(2)}=\left(
\begin{array}{c|c|c}
w_{22}L_2^t     & \cdots & w_{m2}L_2^t \\ \hline
w_{23}L_3^t &  \cdots & w_{m2}L_3^t \\ \hline
\vdots &  \ddots & \vdots \\ \hline
w_{2m}L_m^t  & \cdots & w_{mm}L_m^t
\end{array}
\right).$
%$$
%\left(Id-
%\begin{array}{c|c|c}
%w_{22}L_2^t     & \cdots & w_{m2}L_2^t \\ \hline
%w_{23}L_3^t &  \cdots & w_{m2}L_3^t \\ \hline
%\vdots &  \ddots & \vdots \\ \hline
%w_{2m}L_m^t  & \cdots & w_{mm}L_m^t
%\end{array}
%\right)
%(Id-w_{11}L_1^t)-\left(
%\begin{array}{c}
%w_{12}L_2^t \\ \hline
%w_{13}L_3^t\\ \hline
%\vdots  \\ \hline
%w_{1m}L_m^t
%\end{array}
%\right)
%(w_{21}L_1,\dots, w_{m1}L_1)
%$$
Once the $Q_i's$ are obtained we use
$$
p_1=(w_{21}L_1,\dots, w_{m1}L_1)\left(
\begin{array}{c}
Q_2\\ \hline
Q_3\\ \hline
\vdots  \\ \hline
Q_m
\end{array}
\right)
$$
to get $p_1$. The remaining $p_i's$ are analogously obtained.

\noindent\underbar{Block matrix of type $\B_3$}: The obtention of the Perron vector $\pi$ of $\B_3$  follows from combining the $p_i's$ with the coupling factor, which is the Perron vector of $W^t$.

\[
\B_3=
\left(
\begin{array}{c|c|c|c}
w_{11}L_1^t    & w_{21}Id & \cdots & w_{m1}Id \\ \hline
w_{12}Id & w_{22}L_2^t    & \cdots & w_{m2}Id \\ \hline
\vdots & \vdots & \ddots & \vdots \\ \hline
w_{1m}Id & w_{2m}Id & \cdots & w_{mm}L_m^t
\end{array}
\right)\in \mathbb{R}^{nm\times nm}.
\]

In this case the Perron vector $p_1$ of the Perron complement $P_{11}$  satisfies
$$
p_1=w_{11}L_1^t p_1+(w_{21}Id\dots w_{m1}Id)
\left(Id-\left(
\begin{array}{c|c|c|c}
w_{22}L_2^t & w_{32}Id & \cdots & w_{m2}Id \\ \hline
w_{23}Id & \cdots& \cdots & w_{m2}L_m^t \\ \hline
\vdots &  \vdots & \ddots	 &\vdots \\ \hline
w_{2m}Id  & w_{3m}Id& \cdots & w_{mm}L_m^t
\end{array}
\right)\right)^{-1}
\left(
\begin{array}{c}
w_{12}Id \\ \hline
w_{13}Id\\ \hline
\vdots  \\ \hline
w_{1m}Id
\end{array}
\right)p_1
$$
Then, multiplying in both sides by
$
\left(
\begin{array}{c}
w_{12}Id \\ \hline
w_{13}Id\\ \hline
\vdots  \\ \hline
w_{1m}Id
\end{array}
\right)
$
and using the change of variables

$$
\left(
\begin{array}{c}
Q_2 \\ \hline
Q_3\\ \hline
\vdots  \\ \hline
Q_m
\end{array}
\right)=\left(Id-\left(
\begin{array}{c|c|c|c}
w_{22}L_2^t & w_{32}Id & \cdots & w_{m2}Id \\ \hline
w_{23}Id & \cdots& \cdots & w_{m2}L_m^t \\ \hline
\vdots &  \vdots & \ddots	 &\vdots \\ \hline
w_{2m}Id  & w_{3m}Id& \cdots & w_{mm}L_m^t
\end{array}
\right)\right)^{(-1)}
\left(
\begin{array}{c}
w_{12}Id \\ \hline
w_{13}Id\\ \hline
\vdots  \\ \hline
w_{1m}Id
\end{array}
\right)p_1
$$ so

$\left(
\begin{array}{c}
w_{12}Id \\ \hline
w_{13}Id\\ \hline
\vdots  \\ \hline
w_{1m}Id
\end{array}
\right)p_1
=w_{11}\left(\begin{array}{c}
w_{12}Id \\ \hline
w_{13}Id\\ \hline
\vdots  \\ \hline
w_{1m}Id
\end{array}
\right)
L_1^tp_1+
\left(
\begin{array}{c}
w_{12}Id \\ \hline
w_{13}Id\\ \hline
\vdots  \\ \hline
w_{1m}Id
\end{array}
\right)(w_{21}Id\dots w_{m1}Id)\left(
\begin{array}{c}
Q_2\\ \hline
Q_3\\ \hline
\vdots  \\ \hline
Q_m
\end{array}
\right)
$
or
$\left(
\begin{array}{c}
w_{12}Id \\ \hline
w_{13}Id\\ \hline
\vdots  \\ \hline
w_{1m}Id
\end{array}
\right)p_1
=w_{11}
\left(Id-
\begin{array}{c|c|c|c}
L_1^t  &0   & \cdots & 0 \\ \hline
0 &  L_1^t& \cdots & 0 \\ \hline
\vdots & \vdots & \ddots & \vdots \\ \hline
0 & 0 & \cdots & L_1^t
\end{array}
\right)
\left(\begin{array}{c}
w_{12}Id \\ \hline
w_{13}Id\\ \hline
\vdots  \\ \hline
w_{1m}Id
\end{array}
\right)
L_1^tp_1+
\left(
\begin{array}{c}
w_{12}Id \\ \hline
w_{13}Id\\ \hline
\vdots  \\ \hline
w_{1m}Id
\end{array}
\right)(w_{21}Id\dots w_{m1}Id)\left(
\begin{array}{c}
Q_2\\ \hline
Q_3\\ \hline
\vdots  \\ \hline
Q_m
\end{array}
\right).
$
This is equivalent, by the change of variables above,  
to
 $\left(
\begin{array}{c}
Q_2 \\ \hline
Q_3\\ \hline
\vdots  \\ \hline
Q_m
\end{array}
\right)$ being an eigenvector associated to 1 of the matrix
$${\mathcal A_3}^{p_1}=w_{11}L+\tilde{W}_{11}^{(3)}-w_{11}L\tilde{W}_{11}^{(3)} +\left(\begin{array}{c}
w_{12}Id \\ \hline
w_{13}Id\\ \hline
\vdots  \\ \hline
w_{1m}Id
\end{array}
\right)(w_{21}Id\dots w_{m1}Id)
$$ where
$L=\left(
\begin{array}{c|c|c}
L_1^t  & \cdots & 0 \\ \hline
0 &L_1^t &  \cdots  \\ \hline
\vdots &  \ddots & \vdots \\ \hline
0 &\cdots & L_1^t
\end{array}
\right)
$
and
$\tilde{W}_{11}^{(3)}=\left(
\begin{array}{c|c|c|c}
w_{22}L_2^t & w_{32}Id   & \cdots & w_{m2}Id \\ \hline
w_{23}Id &w_{33}L_3^t &  \cdots & w_{m2}Id \\ \hline
\vdots & \vdots & \ddots & \vdots \\ \hline
w_{2m}Id & w_{3m}Id & \cdots & w_{mm}L_m^t
\end{array}
\right)$

%to $(Q_2,\dots,Q_m)$ belonging to the Kernel of
%
%$$Id-w_{11}L-\tilde{W}+w_{11}L\tilde{W} -\left(\begin{array}{c}
%w_{12}Id \\ \hline
%w_{13}Id\\ \hline
%\vdots  \\ \hline
%w_{1m}Id
%\end{array}
%\right)(w_{21}Id\dots w_{m1}Id)
%$$
%
%where $L=\left(
%\begin{array}{c|c|c}
%L_1^t  & \cdots & 0 \\ \hline
%0 &L_1^t &  \cdots  \\ \hline
%\vdots &  \ddots & \vdots \\ \hline
%0 &\cdots & L_1^t
%\end{array}
%\right)
%$
%and
%$\tilde{W}=\left(
%\begin{array}{c|c|c|c}
%w_{22}L_2^t & w_{32}Id   & \cdots & w_{m2}Id \\ \hline
%w_{23}Id &w_{33}L_3^t &  \cdots & w_{m2}Id \\ \hline
%\vdots & \vdots & \ddots & \vdots \\ \hline
%w_{2m}Id & w_{3m}Id & \cdots & w_{mm}L_m^t
%\end{array}
%\right)$
%\vspace{3mm}

Once the $Q_i's$ are obtained we use the change of variables above to recover $p_1$ (since some of the $w_1i\ne0$):
$$
\left(
\begin{array}{c}
w_{12}Id \\ \hline
w_{13}Id\\ \hline
\vdots  \\ \hline
w_{1m}Id
\end{array}
\right)p_1=\left(Id-\left(
\begin{array}{c|c|c|c}
w_{22}L_2^t & w_{32}Id & \cdots & w_{m2}Id \\ \hline
w_{23}Id & \cdots& \cdots & w_{m2}L_m^t \\ \hline
\vdots &  \vdots & \ddots	 &\vdots \\ \hline
w_{2m}Id  & w_{3m}Id& \cdots & w_{mm}L_m^t
\end{array}
\right)\right)\left(
\begin{array}{c}
Q_2 \\ \hline
Q_3\\ \hline
\vdots  \\ \hline
Q_m
\end{array}
\right).
$$
The remaining $p_i's$ are analogously calculated.

%Notice that $c_{ij}=\|w_{ji}L_j(p_j)\|_1=w_{ji}\|L_jp_j\|_1$
%
%Now, since  $w_{ji}L_j$ are column stochastic we have $(1,\dots,1)w_{ji}L_j=(1,\dots,1)$; thus
%$$\xi_ip_i=\overset{k}{\underset{j=1}{\sum}}c_{ij}\xi_jp_j=\overset{k}{\underset{j=1}{\sum}}(1,\dots,1)w_{ij}w_{ji}L_jp_j\xi_jp_i=\overset{k}{\underset{j=1}{\sum}}w_{ij}\xi_j(1,\dots,1)p_jp_i=\overset{k}{\underset{j=1}{\sum}}w_{ij}\xi_jp_i$$
%
%This implies that $\xi_i=\overset{k}{\underset{j=1}{\sum}}w_{ij}\xi_j$ as $p_i$ is a non-zero vector. In other words, $(\xi_1,\dots,\xi_k)$ is the Perron eigenvector of $W=(w_{ij})$

\bigskip

\bigskip
\subsection{Particular case of two layers ($m=2$)}  

We will show that the eigenvectors associated to the principal eigenvalue 1 can be computed in terms of the eigenvectors associated to 1 of certain matrices related to $L_1^t$, $L_2^t$ and the elements of $W$. Instead of using the techniques of  \cite{Meyer89} we will do all the calculations directly. Moreover, we will deal with all possible cases of $W$ under the only hypothesis that this matrix is row-stochastic.

\noindent\underbar{Block matrix of type $\B_1$, $m=2$}:

\[
\B_1=
\left(
\begin{array}{c|c}
w_{11}L_1^t    & w_{21}L_2^t  \\ \hline
w_{12}L_1^t & w_{22}L_2^t    \\ %\vdots & \vdots & \ddots & \vdots \\ \hline
%w_{1m}L_1^t & w_{2m}L_2 ^t & \cdots & w_{mm}L_m^t
\end{array}
\right),
\]
where $L_\ell^t$ is the transpose of the row normalization of the adjacency matrix of layer $S_\ell$.

\noindent (a) If both $w_{11}\ne 1$ and $w_{22}\ne 1$ then if $\left(
\begin{array}{c}
\pi_1 \\ \hline
\pi_2
\end{array}
\right)$  is an eigenvector associated to the eigenvalue 1,
we have
$$\left\{
    \begin{array}{ll}
      \pi_1=w_{11}L_1^t \pi_1+w_{21}L_2^t \pi_2,  \\
      \pi_2=w_{12}L_1^t \pi_1+ w_{22}L_2^t \pi_2.
    \end{array}
  \right.
$$ From here,
taking into account that both $(I-w_{11}L_1^t)$ and $(I-w_{22}L_2^t)$ are invertible matrices, we get that
$\pi_1=w_{21}(I-w_{11}L_1^t)^{-1}L_2^t \pi_2$, and  $\pi_2=w_{12}(I-w_{22}L_2^t)^1 L_1^t \pi_1.$ Substituting in the above equations we get
$$\left\{
    \begin{array}{ll}
      \pi_1=(w_{11}L_1^t+w_{12}w_{21}L_2^t(I-w_{22}L_2^t)^{-1}L_1^t) \pi_1,  \\
      \pi_2=(w_{22}L_2^t+w_{12}w_{21}L_1^t(I-w_{11}L_1^t)^{-1}L_2^t)\pi_2.
    \end{array}
  \right.
$$  Now multiplying the first equation by the matrix $(I-w_{22}L_2^t)$ on the left, and the second equation by the matrix $(I-w_{11}L_1^t)$ on the left we get
$$\left\{
    \begin{array}{ll}
      \pi_1=(w_{11}L_1^t+ w_{22}L_2^t+(1-w_{11}-w_{22})L_2^tL_1^t)\pi_1,  \\
      \pi_2=(w_{11}L_1^t+ w_{22}L_2^t+(1-w_{11}-w_{22})L_1^tL_2^t)\pi_2,
    \end{array}
  \right.
$$
i.e.,
$\pi_1$ and $\pi_2$ are eigenvectors associated to 1 to the column stochastic matrices
$$\begin{array}{ll}
      {\mathcal A}_1^{\pi_1}=(w_{11}L_1^t+ w_{22}L_2^t+(1-w_{11}-w_{22})L_2^tL_1^t), \hbox{ and} \\
      {\mathcal A}_1^{\pi_2}=(w_{11}L_1^t+ w_{22}L_2^t+(1-w_{11}-w_{22})L_1^tL_2^t).
    \end{array}
$$

\noindent (b) If $w_{11}=1$ then $w_{12}=0$,
in which case $\B_1$ is of the form,
\[
\B_1=
\left(
\begin{array}{c|c}
L_1^t    & w_{21}L_2^t  \\ \hline
0 & w_{22}L_2^t    \\ %\vdots & \vdots & \ddots & \vdots \\ \hline
%w_{1m}L_1^t & w_{2m}L_2 ^t & \cdots & w_{mm}L_m^t
\end{array}
\right),
\]
and if the vector
$\left(
\begin{array}{c}
\pi_1 \\ \hline
\pi_2
\end{array}
\right)$ is associated to the eigenvalue 1 then
$$\left\{
    \begin{array}{ll}
      \pi_1=L_1^t \pi_1+w_{21}L_2^t \pi_2,  \\
      \pi_2=w_{22}L_2^t \pi_2.
    \end{array}
  \right.
$$
We have one of the three following situations:

\noindent (b.1) $0<w_{22}<1$: in this case $\pi_2=0$ since $L_2^t$ is column stochastic and cannot have nonzero eigenvectors with associated to an eigenvalue $1/w_{22}>1$. Therefore
the eigenvectors associated to 1 of $\B_1$ have the form
$\left(
\begin{array}{c}
\pi_1 \\ \hline
0
\end{array}
\right)$
where $\pi_1$ is an eigenvector of $L_1^t$ associated to 1.

\noindent (b.2) $w_{22}=0$: in this case $w_{21}=1$ and we have that
the eigenvectors associated to $\B_1$ have the form
$\left(
\begin{array}{c}
\pi_1 \\ \hline
0
\end{array}
\right)$
where $\pi_1$ is an eigenvector of $L_1^t$ associated to 1.

\noindent (b.3) $w_{22}=1$:
in this case $W$ is the identity (there is no influence of a layer into another layer) and
the eigenvectors of $\B_1$ associated to 1 have the form  $\left(
\begin{array}{c}
\pi_1 \\ \hline
\pi_2
\end{array}
\right)$ where $\pi_1$ is an eigenvector of $L_1^t$ associated to 1 and $\pi_2$ an eigenvector of $L_2^t$ associated to 1.

\noindent (c) If $w_{22}=1$ then, arguing as in case (b) either $w_{11}=1$ and we are again in the situation of (b.3) or the  eigenvector of $\B_1$ associated to 1 are of the form $\left(
\begin{array}{c}
0    \\ \hline
\pi_2
\end{array}
\right)$ where $\pi_2$ is an eigenvector of $L_2^t$ associated to 1.

\noindent\underbar{Block matrix of type $\B_2$, $m=2$}:

\[
\B_2=
\left(
\begin{array}{c|c}
w_{11}L_1^t    & w_{21}L_1^t  \\ \hline
w_{12}L_2^t & w_{22}L_2^t    \\ %\vdots & \vdots & \ddots & \vdots \\ \hline
%w_{1m}L_1^t & w_{2m}L_2 ^t & \cdots & w_{mm}L_m^t
\end{array}
\right),
\]
where $L_\ell^t$ is the transpose of the row normalization of the adjacency matrix of layer $S_\ell$.

\noindent (a) If both $w_{11}\ne 1$ and $w_{22}\ne 1$ then if $\left(
\begin{array}{c}
\pi_1 \\ \hline
\pi_2
\end{array}
\right)$  is an eigenvector associated to the eigenvalue 1,
we have
$$\left\{
    \begin{array}{ll}
      \pi_1=w_{11}L_1^t \pi_1+w_{21}L_1^t \pi_2,  \\
      \pi_2=w_{12}L_2^t \pi_1+ w_{22}L_2^t \pi_2.
    \end{array}
  \right.
$$ From here,
taking into account that both $(I-w_{11}L_1^t)$ and $(I-w_{22}L_2^t)$ are invertible matrices,
we get that $\pi_1=w_{21}(I-w_{11}L_1^t)^{-1}L_1^t \pi_2$, and  $\pi_2=w_{12}(I-w_{22}L_2^t)^{-1} L_2^t \pi_1.$ Substituting in the above equations we get
$$\left\{
    \begin{array}{ll}
     (I-w_{11}L_1^t) \pi_1=w_{12}w_{21}L_1^t(I-w_{22}L_2^t)^{-1}L_2^t \pi_1,  \\
     (I-w_{22}L_2^t) \pi_2=w_{12}w_{21}L_2^t(I-w_{11}L_1^t)^{-1}L_1^t \pi_2,
    \end{array}
  \right.
$$
so using that $L_1^t$ and $(I-w_{11}L_1^t)^{-1}$ commute and $L_2^t$ and $(I-w_{22}L_2^t)^{-1}$ commute we have
$$\left\{
    \begin{array}{ll}
     \pi_1=w_{12}w_{21}L_1^t(I-w_{11}L_1^t)^{-1} (I-w_{22}L_2^t)^{-1}L_2^t \pi_1,  \\
     \pi_2=w_{12}w_{21}L_2^t(I-w_{22}L_2^t)^{-1}(I-w_{11}L_1^t)^{-1}L_1^t\pi_2.
    \end{array}
  \right. \eqno{(1)}
$$
Let us define $\pi_1^{aux}=(I-w_{11}L_1^t)^{-1} (I-w_{22}L_2^t)^{-1}L_2^t \pi_1$ and $\pi_2^{aux}=(I-w_{22}L_2^t)^{-1}(I-w_{11}L_1^t)^{-1}L_1^t\pi_2$.
. By the   equations (1)
$$\left\{
    \begin{array}{ll}
     \pi_1=w_{12}w_{21}L_1^t \pi_1^{aux},  \\
     \pi_2=w_{12}w_{21}L_2^t \pi_2^{aux},
    \end{array}
  \right.\eqno{(2)}
$$
and from (1) and (2)
$$\left\{
    \begin{array}{ll}
     (I-w_{22}L_2^t)(I-w_{11}L_1^t) \pi_1^{aux}=L_2^t \pi_1= L_2^tw_{12}w_{21}L_1^t \pi_1^{aux},  \\
     (I-w_{22}L_2^t)(I-w_{11}L_1^t) \pi_2^{aux}=L_1^t \pi_2= L_1^t w_{12}w_{21}L_2^t \pi_2^{aux},
     \end{array}
  \right.
$$
i.e.,
$\pi_1^{aux}$ and $\pi_2^{aux}$ are eigenvectors associated to 1 of the column stochastic matrices
$$\begin{array}{ll}
      {\mathcal A}_2^{\pi_1^{aux}}=(w_{11}L_1^t+w_{22}L_2^t-w_{11}w_{22}L_2^tL_1^t+w_{12}w_{21}L_2^tL_1^t), \hbox{ and} \\
      {\mathcal A}_2^{\pi_2^{aux}}=(w_{11}L_1^t+w_{22}L_2^t-w_{11}w_{22}L_1^tL_2^t+w_{12}w_{21}L_1^tL_2^t).
    \end{array}
$$ After computing $\pi_1^{aux}$ and $\pi_2^{aux}$,
$$\left\{
    \begin{array}{ll}
     \pi_1=w_{12}w_{21}L_1^t \pi_1^{aux},  \\
     \pi_2=w_{12}w_{21}L_2^t \pi_2^{aux}.
    \end{array}
  \right.
$$

\noindent (b) ($w_{11}=1$) and (c)  ($w_{22}=1$) give the same results as for matrices of type $\B_1$.

\noindent\underbar{Block matrix of type $\B_3$, $m=2$}:

\[
\B_2=
\left(
\begin{array}{c|c}
w_{11}L_1^t    & w_{21}I_2  \\ \hline
w_{12}I_2 & w_{22}L_2^t    \\ %\vdots & \vdots & \ddots & \vdots \\ \hline
%w_{1m}L_1^t & w_{2m}L_2 ^t & \cdots & w_{mm}L_m^t
\end{array}
\right),
\]
where $L_\ell^t$ is the transpose of the row normalization of the adjacency matrix of layer $S_\ell$.

\noindent (a) If both $w_{11}\ne 1$ and $w_{22}\ne 1$ then if $\left(
\begin{array}{c}
\pi_1 \\ \hline
\pi_2
\end{array}
\right)$  is an eigenvector associated to the eigenvalue 1,
we have
$$\left\{
    \begin{array}{ll}
      \pi_1=w_{11}L_1^t \pi_1+w_{21} \pi_2,  \\
      \pi_2=w_{12} \pi_1+ w_{22}L_2^t \pi_2.
    \end{array}
  \right.
$$
Taking into account that both $(I-w_{11}L_1^t)$ and $(I-w_{22}L_2^t)$ are invertible matrices,
we get that
$\pi_1=w_{21}(I-w_{11}L_1^t)^{-1} \pi_2$, and  $\pi_2=w_{12}(I-w_{22}L_2^t)^{-1}  \pi_1.$ Substituting in the above equations we get
$$\left\{
    \begin{array}{ll}
     (I-w_{11}L_1^t) \pi_1=w_{12}w_{21}(I-w_{22}L_2^t)^{-1} \pi_1,  \\
     (I-w_{22}L_2^t) \pi_2=w_{12}w_{21}(I-w_{11}L_1^t)^{-1} \pi_2,
    \end{array}
  \right.
$$
so multiplying in both sides by $(I-w_{11}L_1^t)$ and $(I-w_{22}L_2^t)$ respectively  we have
$$\left\{
    \begin{array}{ll}
     w_{12}w_{21}\pi_1=(I-w_{22}L_2^t)(I-w_{11}L_1^t)\pi_1,  \\
     w_{12}w_{21}\pi_2=(I-w_{11}L_1^t)(I-w_{22}L_2^t)\pi_2.
    \end{array}
  \right.
$$
Therefore,
 $\pi_1$ and $\pi_2$ are eigenvectors associated to 1 to the column stochastic matrices
$$
\begin{array}{ll}
      {\mathcal A}_2^{\pi_1}=(w_{11}L_1^t+w_{22}L_2^t-w_{11}w_{22}L_2^tL_1^t+ w_{12}w_{21}I_2), \hbox{ and} \\
      {\mathcal A}_2^{\pi_2}=(w_{11}L_1^t+w_{22}L_2^t-w_{11}w_{22}L_1^tL_2^t+w_{12}w_{21}I_2).
    \end{array}
$$

\noindent (b) ($w_{11}=1$) and (c)  ($w_{22}=1$) give the same results as for matrices of type $\B_1$.

% Final del documento de introduccion 

%%%%%%%%%%%%%%%%%%%%%%%%%%%%%%%%%%%%%%%%%%%%%%%%%%%%%%%%%%%%%%%%%%%%%
%%%%%%%%%%%%%%%%%%%%%%%%%%%%%%%%%%%%%%%%%%%%%%%%%%%%%%%%%%%%%%%%%%%%%

\begin{thebibliography}{99}
\bibitem{CT08}
 {\sc H.\,Civciv, R.\,Turkmen},
 {\it On new version of strong Hadamard exponential function},
 Selçuk J. Appl. Math. {\bf 9} (2008), no. 1, 11-21.


\bibitem{DSCKMPGA}
 {\sc M.\,De Domenico, A.\,Sol\`e-Ribalta, E.\,Cozzo, M.\,Kivel\"a, Y.\,Moreno, M.A.\,Porter, S.\,G\'omez and A.\,Arenas},
 {\it Mathematical Formulation of Multi-Layer Networks},
 Phys. Rev. X {\bf 3}, 041022  (2013).
 
\bibitem{DSGA13}
 {\sc M.\,De Domenico, A.\,Sol\`e-Ribalta, S.\,G\'omez and A.\,Arenas},
 {\it Random Walks on Multiplex Networks},
 arXiv:1306.0519 (2013).


\bibitem{DSGA14}
 {\sc M.\,De Domenico, A.\,Sol\'e-Ribalta, S.\,G\'omez and A.\,Arenas},
 {\it Navigability of interconnected networks under random failures},
 PNAS {\bf 111} (2014) 8351.

\bibitem{EG14}
{\sc E.\,Estrada, J.\,G\'omez-Garde\~{n}es},
 {\it Communicability reveals a transition to coordinated behavior in multiplex networks},
 Phys. Rev. E  {\bf 89} (2014) 042819.

\bibitem{Gomez13}
 {\sc S.\,G\'omez, A.\,D\'iaz-Guilera, J.\,G\'omez-Garde\~{n}es, J.\,P\'erez-Vicente, Y.\,Moreno and A.\,Arenas},
 {\it Diffusion Dynamics on Multiplex Networks},
 Phys. Rev. Lett. {\bf 110} (2013) 028701.

\bibitem{GK2012}
 {\sc M.\,G\"{u}nther, L.\,Klotz},
 {\it Schur's theorem for a block Hadamard product},
 Linear Algebra Appl. {\bf 437} (2012), no. 3, 948-956.
 
\bibitem{HJ91}
 {\sc R.\,Horn and C.A.\,Johnson},
 {\it  Topics in Matrix Analysis},
 Cambridge University Press, 1991.

\bibitem{HMN91}
 {\sc R.A.\,Horn, R.\,Mathias, Y.\,Nakamura},
 {\it Inequalities for unitarily invariant norms and bilinear matrix products},
 Linear and Multilinear Algebra {\bf 30} (1991), no. 4, 303314.

\bibitem{KhatriRao68}
 {\sc C.G.\,Khatri, C.R.\,Rao},
 {\it Solutions to some functional equations and their applications to characterization of probability distributions},
 Sankhya {\bf 30} (1968) 167180.
 
 \bibitem{LangMey06}
 {\sc A.N.\,Langville and C.D.\,Meyer},
 {\it Google's PageRank and Beyond: The Science of Search Engine Ranks},
 Princeton Univ. Press, Princeton (2006).

 
\bibitem{LN94}
 {\sc W.D.\,Launey, J.\,Seberry},
 {\it The strong Kronecker product},
 Journal of Combinatorial Theory, Series A {\bf 66} (1994), no. 2, 192-213.

\bibitem{Liu99}
 {\sc S.\,Liu},
 {\it Matrix results on the Khatri-Rao and Tracy-Singh products},
 Linear Algebra and Its Applications {\bf 289} (1999), no. 13, 267277.

\bibitem{Liu02}
 {\sc S.\,Liu},
 {\it Several inequalities involving Khatri-Rao products of positive semidefinite matrices},
 Linear Algebra and Its Applications {\bf 354} (2002), no. 13, 175186.

\bibitem{LT08}
 {\sc S.\,Liu, G.\,Trenkler},
 {\it Hadamard, khatri-rao, kronecker and other matrix products},
 International Journal of Information and System Sciences {\bf 4} (2008), no. 1, pp. 160177.

\bibitem{LS82}
 {\sc L.\,Ljung, T.\,S\"{o}derstr\"{o}m},
 {\it  Theory and Practice of Recursive Identification},
 MIT Press, 1982.
 
 \bibitem{Meyer89}
 {\sc C.D.\,Meyer},
 {\it Uncoupling the Perron eigenvector problem},
 Lin.\,Alg.\,Appl., {\bf 114} (1989), 69--94.

\bibitem{Meyer00}
 {\sc C.D.\,Meyer},
 {\it Matrix Analysis and applied linear algebra},
 SIAM, Philadelphia, 2000.


\bibitem{Pitsianis98}
 {\sc N.P.\,Pitsianis},
 {\it Some Properties of the Strong Kronecker Product},
 International Conference on Computational Engineering Science 98, Atlanta, GA, in Modeling and Simulation Based 
 Engineering, Volume I, 433-438, Editors S. N. Alturi and P. E. O Donoghue, Tech Science Press, 1998.

\bibitem{Rao70}
 {\sc C.R\,Rao},
 {\it Estimation of heteroscedastic variances in linear models},
 J. Amer. Statist. Assoc. {\bf 65} (1970) 161172.
 
  \bibitem{Sanco14}
 {\sc R.	J.\,S\'anchez Garc\'{i}a, E.\,Cozzo, Y.\,Moreno},
 {\it Dimensionality reduction and spectral properties of multilayer networks},
 Phys.Rev.E {\bf 89}, 052815.

\bibitem{SRCFGB13}
 {\sc L.\,Sol\'a, M.\,Romance, R.\,Criado, J.\,Flores, A.\,Garc\'{\i}a del Amo and S.\,Boccaletti},
 {\it Eigenvector centrality of nodes in multiplex networks},
 Chaos {\bf 23} (2013) 033131

\bibitem{Schott05}
 {\sc J.F.\,Schott},
 {\it Matrix Analysis for Statistics},
 2nd edition, Wiley, Hoboken, New Jersey, 2005.

\bibitem{SZ93}
 {\sc J.\,Seberry, X-M.\,Zhang},
 {\it Some orthogonal matrices constructed by strong Kronecker product multiplication},
 Austral. J. Combin.{\bf 7} (1993) 213-224.



\bibitem{Sodo13}
 {\sc A.\,Sol\`e-Ribalta, M.\,De Domenico, N.E. Kouvaris, A. Diaz-Guilera, S.\,G\'omez and A.\,Arenas},
 {\it Spectral properties of the laplacian of multiplex networks},
 Phys.Rev.E {\bf 88} (2013) 032807.

\bibitem{XuStoicaLi06}
 {\sc L.\,Xu, P.\,Stoica, J.\,Li},
 {\it  A block-diagonal growth curve model},
 Digital Signal Processing {\bf 16} (2006), no. 6, 902-912.

\bibitem{Zhang04}
 {\sc X.\,Zhang},
 {\it Matrix Analysis and Applications},
 Tsinghua Univ. Press Springer, Beijing, 2004.


%%%%%%%%%%%%%%%%%%%%%%%%%%%%%%%%%%%%%%%%%%%%%%%%%%%%%%%%%%%%%%%%%%%%%
\end{thebibliography}
\end{document}